\documentclass[12pt]{article}

\usepackage{amssymb,amsthm,amsmath,amscd}
\usepackage{url}
\usepackage{graphicx}
\usepackage{subfig}
\usepackage{natbib}

\newtheorem{theorem}{Theorem}[section]
\newtheorem{lemma}[theorem]{Lemma}

\newtheorem{proposition}[theorem]{Proposition}

\theoremstyle{definition} 
\newtheorem*{example}{Example}

\theoremstyle{remark} 
\newtheorem*{remark}{Remark}

\newcommand{\bitem}{\begin{itemize}}
\newcommand{\eitem}{\end{itemize}}
\newcommand{\bcenter}{\begin{center}}
\newcommand{\ecenter}{\end{center}}
\newcommand{\benum}{\begin{enumerate}}
\newcommand{\eenum}{\end{enumerate}}
\newcommand{\beqar}{\begin{eqnarray*}}
\newcommand{\eeqar}{\end{eqnarray*}}
\newcommand{\beq}{\begin{equation}}
\newcommand{\eeq}{\end{equation}}
\newcommand{\bdesc}{\begin{description}}
\newcommand{\edesc}{\end{description}}

\newcommand{\Bin}{{\rm Bin}}
\newcommand{\length}{{\rm length}}

\newcommand{\Uniform}{{\rm Uniform}}
\newcommand{\eps}{\varepsilon}
\newcommand{\graph}{{\rm graph}}
\newcommand{\goto}{\rightarrow}

\newcommand{\Hold}{{\cal H}}

\newcommand{\cF}{{\cal F}}
\newcommand{\cP}{{\cal P}}

\newcommand{\cH}{{\cal H}}

\newcommand{\bbZ}{{\mathbb{Z}}}

\newcommand{\be}{{\bf e}}
\newcommand{\bi}{{\bf i}}

\newcommand{\bh}{{\bf h}}
\newcommand{\bm}{{\bf m}}

\newcommand{\bs}{{\bf s}}

\newcommand{\bx}{{\bf x}}
\newcommand{\by}{{\bf y}}
\newcommand{\bz}{{\bf z}}

\newcommand{\bone}{{\bf 1}}

\newcommand{\pix}{{\rm pixel}}
\def\oR{R_\delta}

\newcommand{\bbG}{\mathbb{G}}
\newcommand{\bbN}{\mathbb{N}}

\newcommand{\bbR}{\mathbb{R}}
\newcommand{\bbB}{\mathbb{B}}

\newcommand{\im}{{\rm im}}

\def\afloor{\lfloor \alpha \rfloor}

\newcommand{\pr}[1]{{\bf P}\left\{#1\right\}}

\newcommand{\QED}{{\unskip\nobreak\hfil\penalty50\hskip2em\vadjust{}
\nobreak\hfil$\Box$\parfillskip=0pt\finalhyphendemerits=0\par}\vspace{.3in}}

\newcommand{\bino}[2]{\left(\begin{array}{c} #1 \\ #2
    \end{array}\right)}

\title{Networks of Polynomial Pieces with Application to the Analysis of Point Clouds and Images
}
\author{Ery Arias-Castro$^1$\footnote{Corresponding author: \tt{eariasca@math.ucsd.edu}}, Boris Efros$^2$ and Ofer Levi$^2$}
\date{
$^1$ University of California, San Diego, USA\\
$^2$ Ben Gurion University of the Negev, Be'er Sheva, Israel\\
}

\begin{document}

\maketitle

\begin{abstract}
We consider H\"older smoothness classes of surfaces for which we construct piecewise polynomial approximation networks, which are graphs with polynomial pieces as nodes and edges between polynomial pieces that are in `good continuation' of each other.
Little known to the community, a similar construction was used by Kolmogorov and Tikhomirov in their proof of their celebrated entropy results for H\"older classes. 

We show how to use such networks in the context of detecting geometric objects buried in noise to approximate the scan statistic, yielding an optimization problem akin to the Traveling Salesman.
In the same context, we describe an alternative approach based on computing the longest path in the network after appropriate thresholding. 

For the special case of curves, we also formalize the notion of `good continuation' between beamlets in any dimension, obtaining more economical piecewise linear approximation networks for curves.

We include some numerical experiments illustrating the use of the beamlet network in characterizing the filamentarity content of 3D datasets, and show that even a rudimentary notion of good continuity may bring substantial improvement.       
\end{abstract}


\noindent {\it Acknowledgments:} The authors are grateful to the anonymous referees and editors for their careful reading of an earlier version of the manuscript and their constructive comments. 
The first author was partially supported by NSF grant DMS-0603890.  

\section{Introduction}

\subsection{From function approximation to set approximation}

An important trend in Approximation Theory and Harmonic Analysis focuses on designing dictionaries (e.g. orthonormal bases) $\{\phi_n\}$ well-adapted to a given function class $\cF$, in the sense that any $f \in \cF$ is well-approximated by a linear combination of a few functions from the dictionary:
\begin{equation} \label{eq:f-decomposition}
f \doteq \sum_{n \in N(f)} a_n \phi_n.
\end{equation} 
Examples of dictionaries include the Fourier basis, polynomials, splines \cite{MR1045442}, radial basis functions, wavelets \cite{MR1614527} and others such as wedgelets \cite{wedgelets}, platelets \cite{1199635}, bandelets \cite{MR2128287}, rigelets \cite{MR1857974}, curvelets \cite{MR2012649}, chirplets \cite{chirp-tech}.  

In the context of function estimation in additive white noise \cite{johnstone1998feg, MR1635414, MR1256530}, approximations by sums of atoms as in (\ref{eq:f-decomposition}) are particularly suitable.
Consider instead a setting where a geometric object (i.e. a set) is buried in noise or clutter, a setting considered e.g. in \cite{2D-beamlets, MGD}.
This is for example relevant in target tracking, where the object of interest is the target's trajectory often modeled as a curve.
In this setting, set approximation plays the role of function approximation, and the aim of the present paper is to develop strategies to compute approximations of sets by unions of simple building blocks akin to how beamlets are used to approximate curves in \cite{2D-beamlets}.

Since a set may be equivalently represented by its indicator function, approximation of sets may appear to be a special case of approximation of functions.
Though indeed closely related, function approximation does not directly translate into set approximation.
In part, this comes from the fact that the image of a sum of functions is in general not the union of the functions' images.
However, when a parametrization of the set is available, approximating the parametrization in piecewise fashion (i.e.  when the supports of the functions involved in the sum do not overlap) does result in a proper approximation of the set itself.
This is for example the case when the set is the graph of a function, considered in \cite{kor-tsy, wedgelets} in the context of image processing.

The selection of building blocks in set approximation is also not necessarily parallel to that of atoms in function approximation.
The main difference is in the fact that overlapping building blocks create redundancy, while atoms with overlapping supports may cancel each other in appropriate ways to fit the target function, e.g. in trigonometric, polynomial or wavelet expansions. 
Typically, the approximation (\ref{eq:f-decomposition}) is constructed by first computing the coefficients $a_n$, often as the inner product of $f$ and $\phi_n$, and then keeping the largest ones in absolute value, which is in effect equivalent to thresholding \cite{MR1256530}.
A similar strategy may be implemented for sets, where building blocks with the most overlap with the set (as a fraction of their size) are selected; this corresponds to the simplest beamlet-based algorithm presented in \cite{2D-beamlets}.
However, the result is often redundant, as nearby building blocks tend to have a similar overlap with a given set.
This is avoided in the beamlet-based coder for curves introduced in \cite{1518933}, by choosing at most one beamlet per dyadic square in a a given recursive dyadic partitioning. 
The strategy we adopt here consists of looking for a union of building blocks that obey some sort of `good continuation'.
In the approximation of curves, this corresponds to the most complex algorithm presented in \cite{2D-beamlets} which amounts to chaining beamlets together that share one endpoint and have similar orientations.
This is formalized in \cite{MSDFS} for curves that are graphs of H\"older functions and in \cite{MR2379113} for chirps (highly oscillating functions), chaining chirplets in good continuation. 
Good-continuation principles are also considered in \cite{766676,MR2036391} inspired by the Gestalt theory of vision, with applications to the detection of (parametric) geometric object in images.

\subsection{Networks of polynomial pieces}

We define a network of polynomial pieces to be a graph with nodes indexing polynomial pieces and edges between polynomial pieces in good-continuation.
(We use the word `network' instead of `graph' so as to avoid confusion with the notion of `graph of a function'.)
We construct such networks to approximate surfaces of varying dimension and smoothness.
Just as beamlets are line-segments spanning a wide range of location, orientations and scale, so do the polynomial pieces.
The networks are akin to that of \cite{MSDFS} in that the notion of good-continuation is explicit and the scales are not mixed together, thus better adapted to surfaces with homogeneous smoothness, e.g. graphs of H\"older functions.
We do, however, suggest ways to go multiscale.

In the 1950's, Kolmogorov and Tikhomirov used piecewise polynomial approximations together with a similar kind of good-continuation notion to bound the $\epsilon$-entropy of H\"older function classes \cite{kolmogorov}.
Our construction may be seen as a formalization of their approach, with emphasis on the approximation of sets instead of functions.
Note that this ancestry was discovered after the fact; the present perspective was indeed independently suggested in \cite{ery-thesis} with the intention of generalizing the system used in \cite{MSDFS}.

\subsection{The special case of beamlets}

Beamlets were introduced in \cite{2D-beamlets} for the explicit purpose of approximating curves in 2D, with a 3D version latter developed in \cite{3D-beamlets}.
A variety of algorithms are proposed in \cite{2D-beamlets}, where the more elaborate ones are based on the chaining of beamlets in good continuation.
However, the notion of good-continuation remains implicit and only palpable through numerical experiments.
  
We formalize here this notion of good-continuation for beamlets.
This was previously done in \cite{MSDFS} for a beamlet-like system built to detect graphs of H\"older functions.
In the resulting beamlet network, small beamlets have a large number of neighbors; this seems unavoidable if the network is to enable accurate approximation of curves.
To address this issue, we develop an alternative network of line-segments, with emphasis on developing an economical system both in terms of number of nodes and connectivity.

\subsection{Application to the analysis of point clouds and images}

\subsubsection{Detection of geometric objects}

Consider a simple model for detection, where we observe a point cloud (i.e. a spatial point process) and the goal is to decide whether the points were generated uniformly at random or a fraction of them were sampled from a geometric object (i.e. a surface) belonging to a given class.

This is a standard setting where the scan statistic is used \cite{GlaNauWal,GlaBal}, which consists in computing the largest number of points belonging to one of the objects of interest.
Though the scan statistic achieves the best known detection rates, both for parametric \cite{MGD} and nonparametric \cite{CTD} classes of objects, it is not computationally friendly as it involves an optimization over a large (function) class of objects, though there are exceptions \cite{CTD-DP}.
Instead, we propose to replace the optimization over the class of objects with an optimization over paths in an approximating network.
The idea of replacing an optimization over a function space with an optimization over a carefully constructed graph is quite natural, and in fact appears in other situations, e.g. in the computation of minimal surfaces \cite{1238310,10.1109/ICCV.2005.252,kirsanov2004dgm,Hu_optimalminimum-surface}.
In our particular case, the resulting optimization problem is  akin to the Budget-Reward problem \cite{DasHesSon}; though still NP-hard, this problem admits polynomial time approximations.
Moreover, we show that this approximation achieves the best known detection rates.

Note that the optimization above is computationally more tractable (e.g. using dynamic programming ideas) when the network is direct and acyclic, which is the case for example in multiframe target tracking; see e.g. \cite{IMA_MFD}, where a beamlet network is used to track the time-space primitives.

We also consider an alternative approach based on computing the size of the longest path in the network after thresholding, which is suggested in \cite{2D-beamlets}.
Dynamic programming ideas may also be implemented here, as done in \cite{MSDFS}.
We show that this method also achieves the best known detection rates.

We mention that the same approaches may be implemented in the setting of an image, where a geometric object is buried  in white Gaussian noise.

\subsubsection{Characterization of spatial distributions}

In Astrophysics, the study of the galaxy distribution involves  quantifying the content in filaments, sheets and blobs in 3D galaxy catalogs \cite{MarSaa, sdss}.
Ultimately, scientists would like to know which of the many cosmological models is best (and well) supported by observations.
Practically, the task is to meaningfully compare simulated galaxy distributions from various cosmological models with the observed galaxy distribution.

With the presence of highly anisotropic features such as filaments,  traditional tools for analyzing spatial data become irrelevant, among them the classical two-point correlation function.
Instead, a method based on beamlets is very attractive, as beamlets provide good approximations for filaments.
And indeed, \cite{spie-beamlets} presents evidence that beamlets are useful at separating various cosmological models, even though the algorithm implemented in \cite{spie-beamlets} is of the simplest kind and in particular does not involve chaining (i.e. good-continuation).
We perform a number of numerical experiments on simulated data that show that chaining may bring substantial improvement. 

Note that such tools are in demand in other scientific fields, such as Medical Imaging, for example in the examination of vascular networks \cite{APPLIC_3D_NETWORK} or cancer cells \cite{APPLIC_FILAMENT_IN_CELLS}.

\subsection{Contents}

The contents are organized as follow.
In Section 2, we introduce piecewise polynomial networks designed to approximate surfaces of any intrinsic dimension and (H\"older) smoothness.
In Section 3, we consider the detection of geometric objects buried in noise and develop methods based on these networks.
In Section 4, we formalize a notion of good-continuation for beamlets in arbitrary dimension and show that the resulting network has desirable approximation properties for curves.
In Section 5, we perform some numerical experiments showing that the notion of good-continuation may bring practical improvement.
Some of the proofs and technical arguments are gathered in the Appendix.

\section{Networks of Polynomial Pieces}
\label{sec:hold}

We build an explicit, algorithmically friendly approximating network for H\"older surfaces, i.e. graphs or images of H\"older functions.
As a H\"older function is well-approximated locally by a polynomial, in fact its Taylor expansion, it is natural to construct piecewise polynomial approximations.
The idea is to partition the unit hypercube into smaller hypercubes and in each smaller hypercube provide a choice of approximation by polynomials; since the functions we consider are uniformly smooth, approximations in nearby hypercubes should be close, which we formalize as a neighboring condition.
We thus form a network with nodes indexing polynomial pieces and edges linking those in good-continuation, restricting the possible combinations to those useful in approximating functions of given smoothness, in such a way that functions and certain connected components in this network are in correspondence.
Though we build a different network for each smoothness class, it is possible to discretize the range of parameters resulting in a dyadic organization of this family of networks by scale.

Little known to the community, a similar construction was used by Kolmogorov and Tikhomirov in their seminal work on $\epsilon$-entropy of H\"older function classes (and others) \cite{MR0124720, kolmogorov}.
Note that the present construction was independently suggested in \cite{ery-thesis}, as a generalization of the beamlet-like system used in \cite{MSDFS}.

We first introduce some notation.
For $i \in \{1,\dots,k\}$, let $\be_i$ denote the $i$th
canonical vector in $\bbR^k$.
For a vector $\bx = (x_1,\dots,x_k) \in \bbR^k$, its supnorm is defined as $\|\bx\| = \max\{|x_i|:i=1,\dots,k\}$.  
For $\bs = (s_1,\dots,s_k) \in \bbN^k$, let $\bs! = s_1! \cdots s_k!$ and $|\bs| = s_1 + \cdots + s_k$.
For a function $f$ and $\bs = (s_1,\dots,s_k) \in \bbN^k$, $f^{(\bs)} = \partial_{x_1}^{s_1} \cdots \partial_{x_k}^{s_k} f$.

We define the following constants:
\begin{equation} \label{eq:c_def} 
c_1 = \sum_{|\bs| = \afloor} \frac{1}{\bs!}, \qquad
c_2 = \sum_{|\bs| \leq \afloor} \frac{2^{-|\bs|}}{\bs!}.
\end{equation}
Note that $c_1 \leq \exp(k)$ and $c_2 \leq \exp(k/2)$.

\subsection{H\"older smoothness classes}

For $\alpha, \beta > 0$, define $\Hold^k(\alpha,\beta)$ as the H\"older smoothness class of $\afloor$-times differentiable functions functions $f:[0,1]^k \goto [0,1]$ satisfying:
\begin{align}
|f^{(\bs)}(\bx)| &\leq \beta, \quad \forall \bx \in [0,1]^k, \ \forall \bs \in \bbN^k, |\bs| \leq \afloor; \label{eq:deriv_bound} \\
|f^{(\bs)}(\by) - f^{(\bs)}(\bx)| &\leq \beta \|\by -\bx\|^{\alpha-\afloor}, \quad \forall \bx, \by \in [0,1]^k, \ \forall \bs \in \bbN^k, |\bs| = \afloor. \label{eq:deriv_max}
\end{align}

\begin{example}{($k=1$, $\alpha \in (1,2]$)}
\begin{align}
|f^{'}(x)| &\leq \beta, \quad \forall x \in [0,1]; \nonumber \\
|f^{'}(y) - f^{'}(x)| &\leq \beta |y -x|^{\alpha-1}, \quad \forall x, y \in [0,1]. \nonumber
\end{align}
\end{example}

Functions in $\Hold^k(\alpha,\beta)$ are uniformly well-approximated locally by polynomials, specifically their Taylor expansions.
For $f \in \Hold^k(\alpha,\beta)$ and $\bx \in [0,1]^k$, the Taylor expansion of $f$ at $\bx$ of degree $\afloor$ is defined as follows:
$$\dot{f}_\bx(\by) = \sum_{|\bs| \leq \afloor} f^{(\bs)}(\bx)  \ \prod_{i=1}^d \frac{(y_i - x_i)^{s_i}}{s_i!}.$$ 

\begin{example}{($k=1$, $\alpha \in (1,2]$)}
$$\dot{f}_x(y) = f(x) + f^{'}(x) (y-x).$$
\end{example} 

\begin{lemma}
\label{lem:taylor}
For any $f \in \Hold^k(\alpha,\beta)$,
$$
|f(\by) - \dot{f}_\bx(\by)| \leq c_1 \beta \|\by - \bx\|^\alpha, \quad \forall \bx, \by \in [0,1]^k.
$$
\end{lemma}

\begin{proof}
A Taylor approximation of degree $\afloor$ gives:
$$f(\by) = \dot{f}_\bx(\by) + \sum_{|\bs| = \afloor} (f^{(\bs)}(\bz) - f^{(\bs)}(\bx))  \ \prod_{i=1}^d \frac{(y_i - x_i)^{s_i}}{s_i!},$$
for some $\bz$ on the segment joining $\bx$ and $\by$.
Hence,
$$|f(\by) - \dot{f}_\bx(\by)| \leq c_1 \|\by - \bx\|^{\afloor} \max_{|\bs| = \afloor} |f^{(\bs)}(\bz) - f^{(\bs)}(\bx)|.$$
Now apply (\ref{eq:deriv_max}) and the fact that $\|\bz - \bx\| \leq \|\by - \bx\|$ to get
$$|f^{(\bs)}(\bz) - f^{(\bs)}(\bx)| \leq \beta \|\by - \bx\|^{\alpha - \afloor}, \quad \forall \bs \in \bbN^k, |\bs| = \afloor.$$

\end{proof}

\subsection{Nets of piecewise polynomials}

We now build a family of nets for $\Hold^{k}(\alpha,\beta)$ by dividing $[0,1]^k$ into hypercubes and then offering a choice of approximation by polynomials within each hypercube, which is most relevant in view of Lemma \ref{lem:taylor}.
Fix $\Delta \in (0,1)$ and $\delta > 0$, and define $\delta_s = \Delta^{-s} \delta, s=0,\dots,\afloor$.  
In the construction that follows, the parameter $\Delta$ quantizes the variable space, while each parameter $\delta_s$ quantizes the range of values of derivatives of order $s$ of functions in $\Hold^{k}(\alpha,\beta)$.
Note that the quantization is coarser for higher order derivatives, and specifically chosen so that the approximation result in Lemma \ref{lem:local_approx} below holds. 
  
Divide $[0,1]^{k}$ into hypercubes indexed by $\bm \in \{1,\dots,\Delta^{-1}\}^k$ of the form:
\begin{eqnarray*}
I_\bm 
&=& \prod_{i=1}^k [(m_i-1) \Delta,m_i \Delta].
\end{eqnarray*}
Let $\bx_\bm = (x_{\bm,1}, \dots, x_{\bm,k})$ denote the center of $I_\bm$, i.e. $x_{\bm,i} = (m_i-1/2) \Delta$.
On each hypercube $I_\bm$, consider polynomials of the form 
\begin{equation}
\label{eq:local_poly}
p_{\bm, \bh}(\bx) = \sum_{|\bs| \leq \afloor} h^{(\bs)} \delta_{|\bs|} \ \prod_{i=1}^d \frac{(x_i - x_{\bm,i})^{s_i}}{s_i!},
\end{equation} 
where $\bh = (h^{(\bs)}: |\bs| \leq \afloor)$, with $h^{(\bs)} \in \bbZ$ and $|h^{(\bs)} \delta_{|\bs|}| \leq \beta$.

\begin{example}{($k=1$, $\alpha \in (1,2]$)} For $\bh = (h^{(0)}, h^{(1)})$,
$$p_{m, \bh}(x) = h^{(0)} \delta + h^{(1)} \delta_1 (x - x_m).$$
\end{example}

\begin{figure}[htbp]
  \centering
\subfloat[Linear]{\includegraphics[height = 2in]{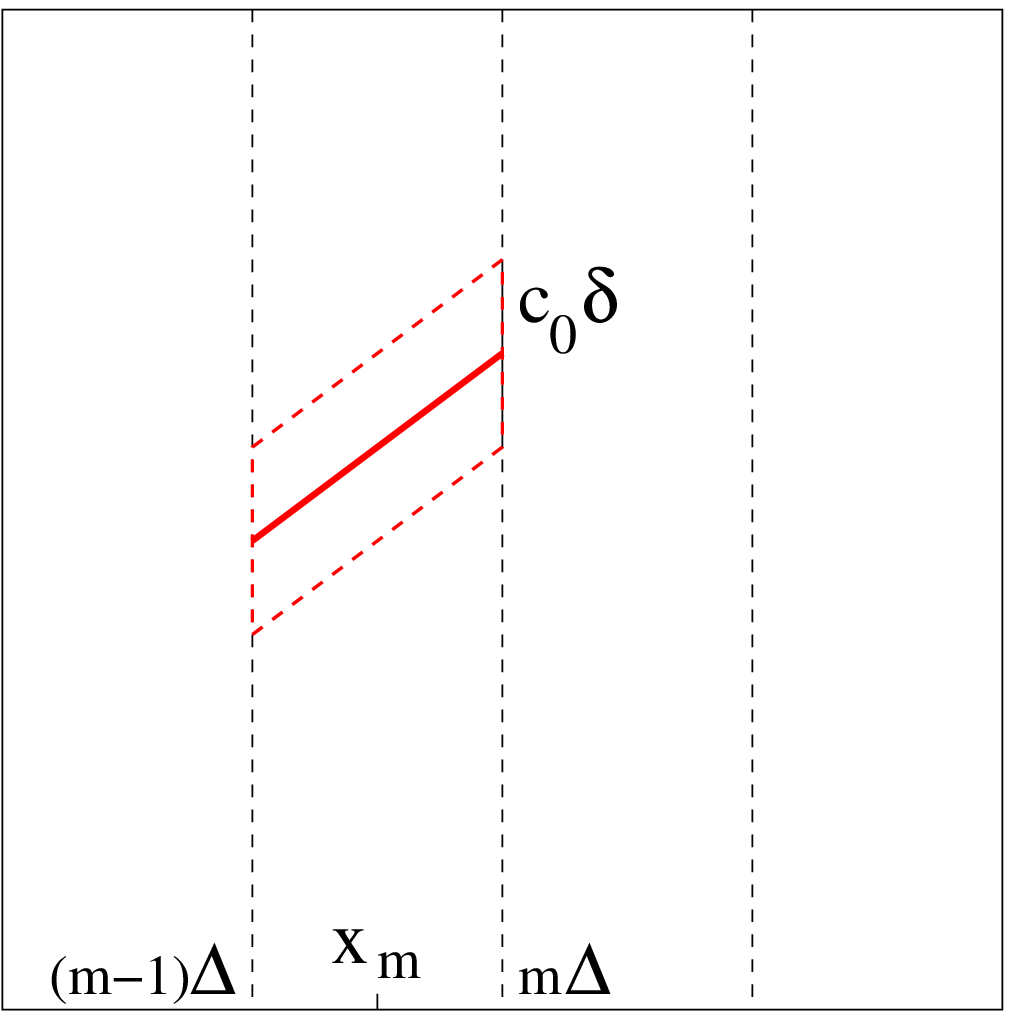}}
  \qquad%
\subfloat[Higher order]{\includegraphics[height = 2in]{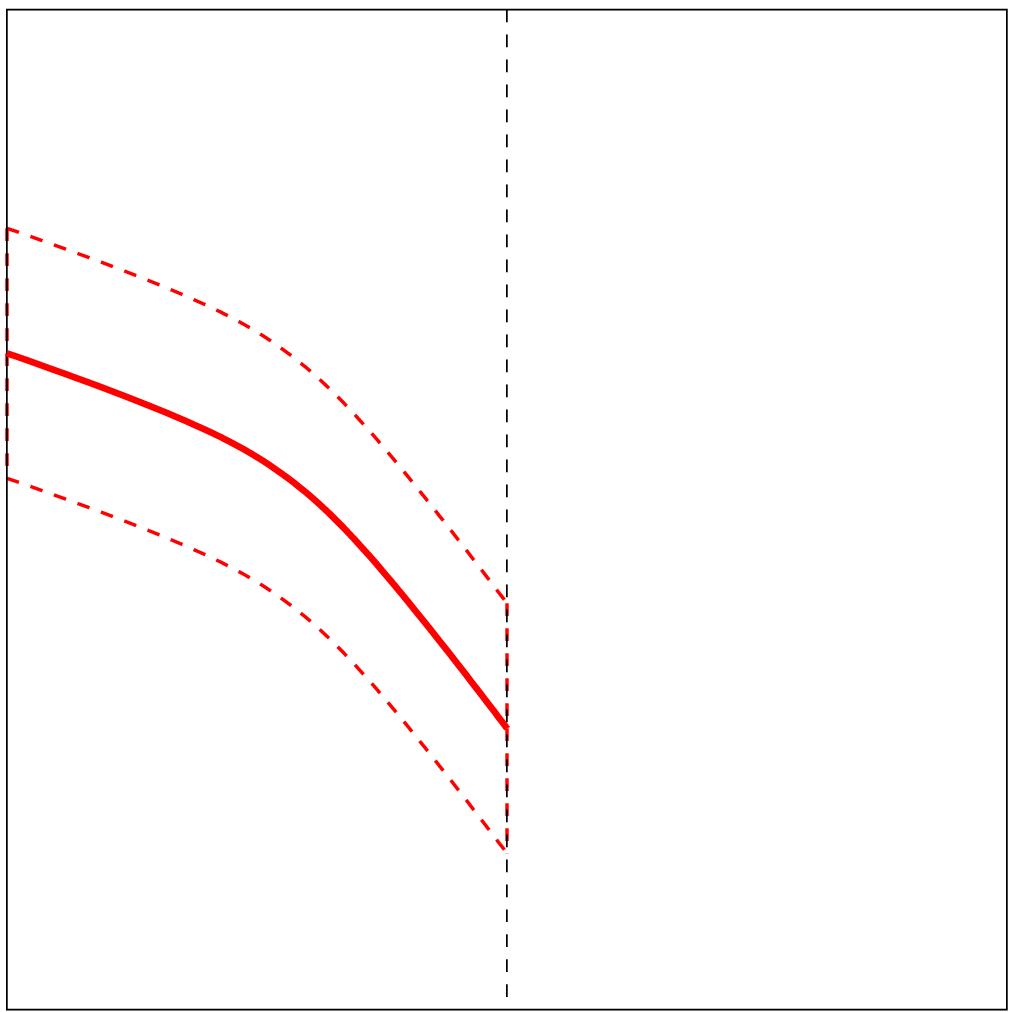}}
  \caption{Examples of polynomial pieces for different choices of $\Delta$ and $\delta$, together with their associated region as defined in (\ref{eq:R}).}
  \label{fig:1D-hold}
\end{figure}

For $f \in \Hold^{k}(\alpha,\beta)$, 
let $h^{(\bs)}(\bm,f)$ denote the closest integer to
$f^{(\bs)}(\bx_\bm)/\delta_{|\bs|}$.

\begin{lemma}
\label{lem:local_approx}
For $f \in \Hold^{k}(\alpha,\beta)$ and $\bm \in \{1,\dots,\Delta^{-1}\}^k$,
$$|f(\bx) - p_{\bm, \bh(\bm,f)}(\bx)| \leq (c_1/2^\alpha) \beta \Delta^\alpha + (c_2/2) \delta, \quad \forall \bx \in I_\bm.$$
\end{lemma}
\begin{proof}
The proof of Lemma \ref{lem:local_approx} is in the Appendix.
\end{proof}

Partly justified by Lemma \ref{lem:local_approx}, we now assume that 
\begin{equation} \label{eq:Delta-delta}
\delta = c_1 \beta \Delta^\alpha.
\end{equation}

The system $\{p_{\bm, \bh}\}$ is therefore rich enough to provide a certain degree of approximation locally.
The same degree of approximation may be achieved globally by simply considering functions that coincide with the polynomials above within each hypercube, namely 
$$g_\bh(\bx) = \sum_{\bm} \bone_{\{\bx \in I_\bm\}} \ p_{\bm, \bh(\bm)}(\bx),$$
where this time $\bh$ also depends on $\bm$, representing a different choice for each hypercube.

As Kolmogorov and Tikhomirov realized, generating a $g_\bh$ by simply picking a polynomial in each hypercube independently would result in a wasteful system, for in fact polynomials in neighboring hypercubes may be restricted to have similar coefficients.
This comes from the fact that the derivatives up to order $\afloor$ of a function in $f \in \Hold^{k}(\alpha,\beta)$ have a certain degree of smoothness.

\begin{lemma}
\label{lem:taylor_ei}
For $f \in \Hold^{k}(\alpha,\beta)$ and $\bs$ such that $|\bs| \leq \afloor$, 
$$\left|f^{(\bs)}(\bx + \eta \be_i) - \sum_{t \leq \afloor - |\bs|} f^{(\bs + t \be_i)}(\bx)\ \frac{\eta^t}{t!}\right| \leq c_1 \beta \eta^{\alpha-|\bs|},$$
for all $\bx \in [0,1]^k$, $i = 1, \dots, k$ and $\eta$ such that $\bx + \eta \be_i \in [0,1]^k$.
\end{lemma}

\begin{proof}
As with Lemma \ref{lem:taylor}, perform a Taylor approximation of degree $\afloor - |\bs|$ along $\be_i$ and apply (\ref{eq:deriv_max}).
In fact, $\beta$ may be replaced by $\beta/(\afloor - |\bs|)!$.
\end{proof}

\subsection{Approximating networks for H\"older graphs}
\label{sec:hold-graph}

For a function $f:[0,1]^k \to [0,1]$, define its graph as
$$\graph(f) = \{(\bx, f(\bx)): \bx \in [0,1]^k\}.$$
Assuming (\ref{eq:Delta-delta}), we define a network of polynomial pieces $\bbG_{\delta}^{k}(\alpha,\beta)$ with the property that a certain kind of connected components index piecewise polynomial approximations for graphs of H\"older functions.
The network $\bbG_{\delta}^{k}(\alpha,\beta)$ has nodes of the form $(\bm,\bh)$ indexing the polynomials $p_{\bm, \bh}$ defined in (\ref{eq:local_poly}).
Two nodes in the network $(\bm,\bh)$ and $(\bm_\star,\bh_\star)$ are neighbors if the corresponding hypercubes, $I_\bm$ and $I_{\bm_\star}$, are adjacent, and if the corresponding polynomials, $p_{\bm, \bh}$ and $p_{\bm_\star, \bh_\star}$, and their derivatives assume nearby values both at $\bx_\bm$ and $\bx_{\bm_\star}$.
Formally, this corresponds to $\bm_\star = \bm + \xi \be_i$ for some $i \in \{1,\dots,k\}$ and $\xi \in \{-1,+1\}$, and for all $\bs \in \bbN^k, |\bs| \leq \afloor$,
\begin{equation} \label{eq:hold_neigh}
\left|h_\star^{(\bs)} - \sum_{t \leq \afloor - |\bs|}
  \frac{\xi^t}{t!} h^{(\bs+ t \be_i)}\right| < 3 \quad {\rm and}
\quad \left|h^{(\bs)} - \sum_{t \leq \afloor - |\bs|}
  \frac{(-\xi)^t}{t!} h_\star^{(\bs+ t \be_i)}\right| < 3.
\end{equation} 
This last property is a discrete version of Lemma \ref{lem:taylor_ei}.

\begin{example}{($k=1$, $\alpha \in (1,2]$)}
The nodes $(m,h^{(0)},h^{(1)})$ and $(m_\star,h_\star^{(0)}, h_\star^{(1)})$ are neighbors if $|m_\star - m| = 1$ and 
$$
|h_\star^{(1)} - h^{(1)}| < 3, \quad 
|h_\star^{(0)} - h^{(0)} - (m_\star - m) h^{(1)}| < 3, \quad 
|h^{(0)} - h_\star^{(0)} + (m_\star - m) h_\star^{(1)}| < 3.
$$
\end{example}

\begin{figure}[htbp]
  \centering
\subfloat[In good continuation]{\includegraphics[height = 2in]{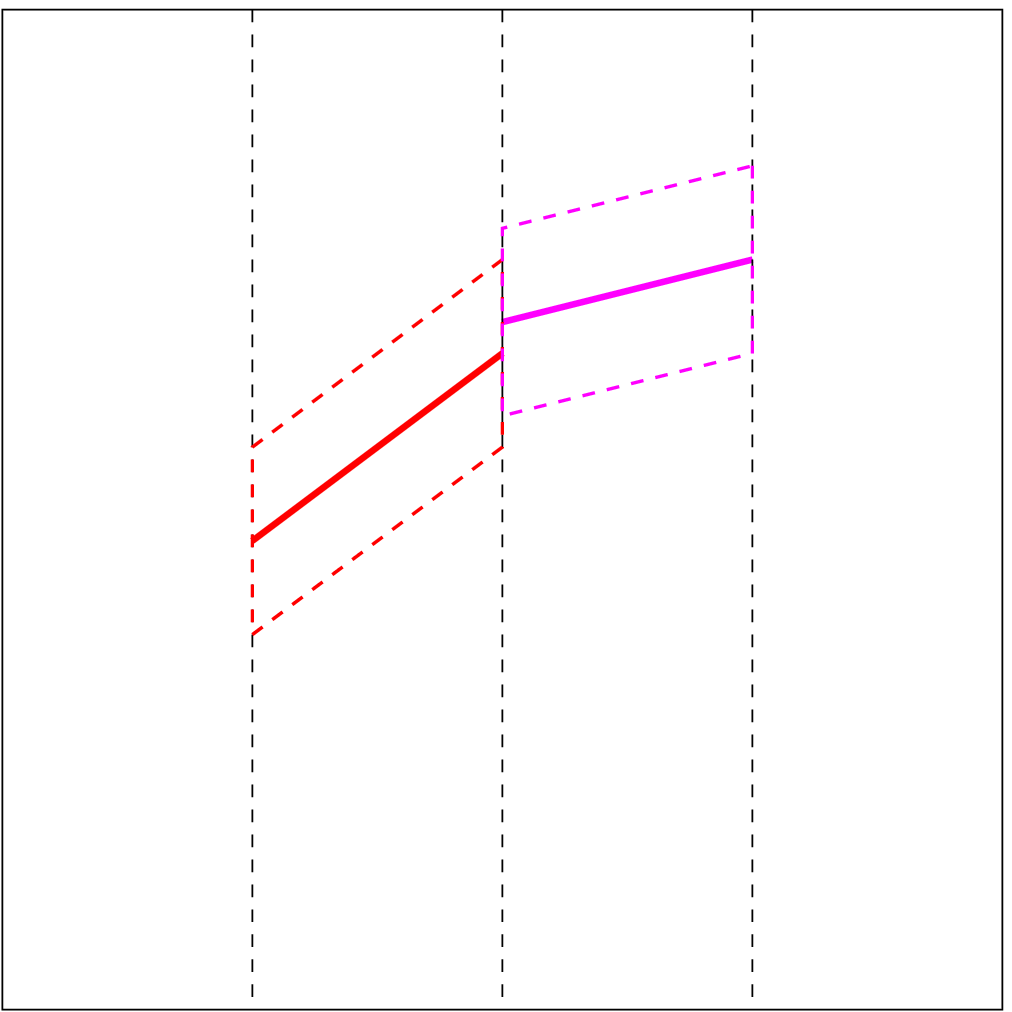}}
  \qquad%
\subfloat[Not in good continuation]{\includegraphics[height = 2in]{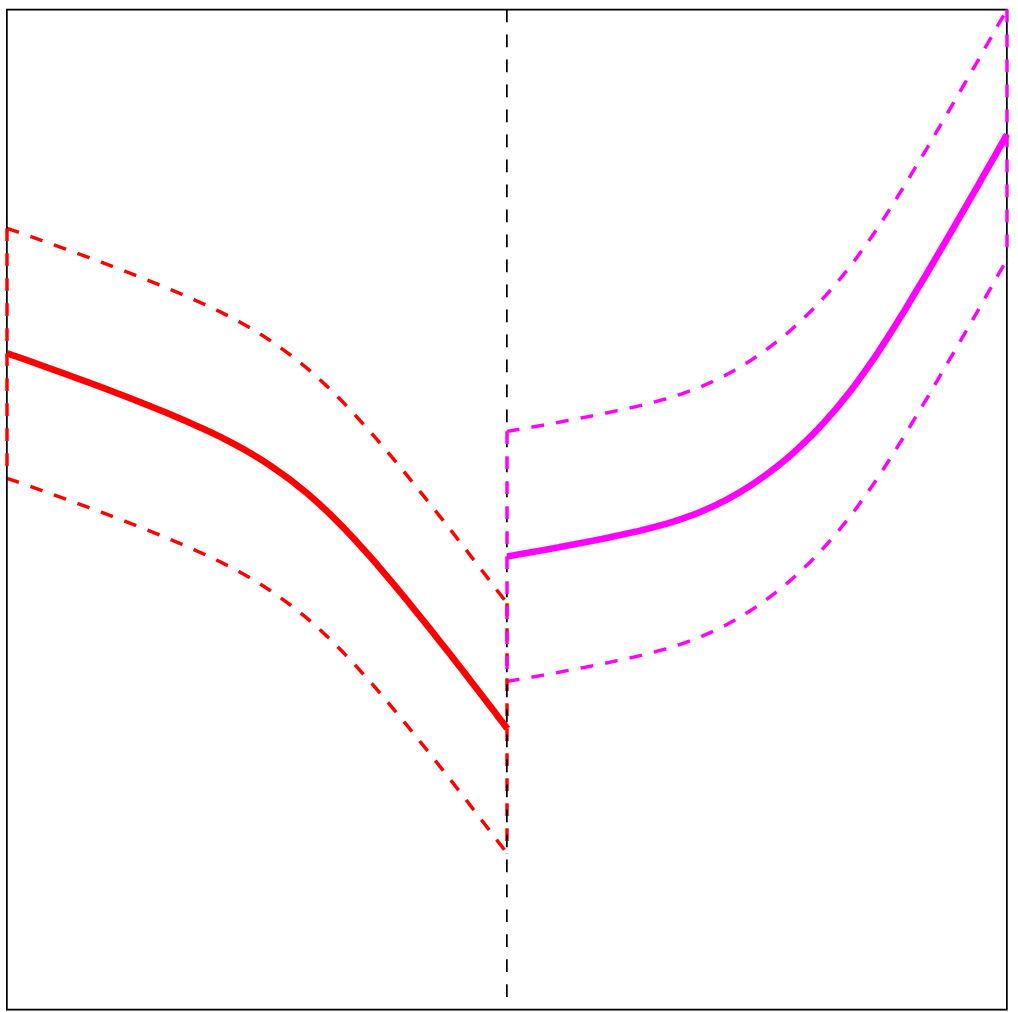}}
  \caption{Examples of polynomial pieces in good continuation (left) and not in good continuation (right).}
  \label{fig:1D-hold-continuation}
\end{figure}

The family of networks $\bbG_{\delta}^{k}(\alpha,\beta)$ effectively generalizes the system presented in \cite{MSDFS} for the case $k=1, \alpha \in (1,2]$.

Define constants
$$c_3 = \sum_{s=0}^{\afloor} \bino{s+k-1}{k-1}, \qquad c_4 = \sum_{s=0}^{\afloor} s \bino{s+k-1}{k-1}.$$

\begin{lemma}
\label{lem:hold-count}
$\bbG_{\delta}^{k}(\alpha,\beta)$ has $O(\beta^{c_3} \delta^{-c_3} \Delta^{c_4-k}) = O(\beta^{c_3-(c_4-k)/\alpha} \delta^{-c_3 + (c_4-k)/\alpha})$ nodes and each node has at most $2k 6^{c_3}$ neighbors.
\end{lemma}
\begin{proof}
The proof of Lemma \ref{lem:hold-count} is in the Appendix.
\end{proof}

\begin{example}{($k=1$, $\alpha \in (1,2]$)}
$\bbG_{\delta}^{1}(\alpha,\beta)$ has $O(\beta \delta^{-2})$ nodes, each with degree at most 72.
\end{example}

To each vertex $(\bm,\bh)$ we associate a region around the graph of $p_{\bm,\bh}$ restricted to $I_\bm$, of thickness given by the error bound of Lemma  \ref{lem:local_approx}:
\begin{equation} \label{eq:R}
R(\bm,\bh) = \{(\bx,z) \in I_\bm \times [0,1]: |z-p_{\bm,\bh}(\bx)| \leq c_0 \delta\},
\end{equation}
where $c_0 = (2^{-\alpha} + c_2/2)$.  See Figure \ref{fig:1D-hold}.

\begin{example}{($k=1$, $\alpha \in (1,2]$)}
The regions are in this case parallelograms.
\end{example}
 
For a subset of nodes $\pi$, define $R(\pi) = \bigcup_{(\bm,\bh) \in \pi} R(\bm,\bh)$.

Let $\Pi_{\delta}^{k}(\alpha,\beta)$ denote the set of connected components of $\bbG_{\delta}^{k}(\alpha,\beta)$, homeomorphic to the square grid $\{1, \dots, \Delta^{-1}\}^k$.

\begin{theorem}
\label{th:graph_cover}
For each $f \in \Hold^{k}(\alpha,\beta)$, there is a connected component $\pi \in \Pi_{\delta}^{k}(\alpha,\beta)$ such that $\graph(f) \subset R(\pi)$.
\end{theorem}
\begin{proof}
The proof of Theorem \ref{th:graph_cover} is in the Appendix.
\end{proof}

\begin{figure}[htbp]
  \centering
\includegraphics[height = 2.5in]{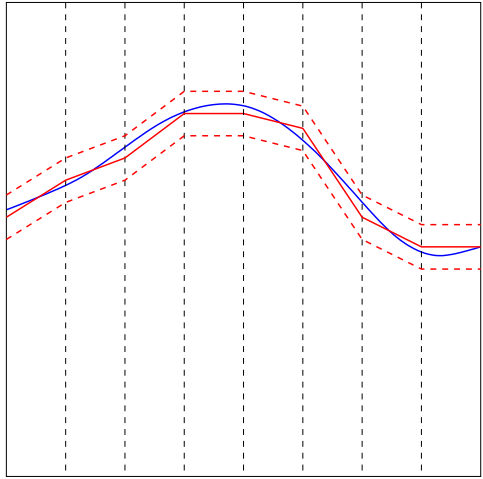}
  \caption{Example of covering of an H\"older graph (blue) with polynomial (linear) pieces in good continuation (red).}
  \label{fig:1D-hold-covering}
\end{figure}

Theorem \ref{th:graph_cover} implies that the system of functions
$$\sum_{(\bm, \bh) \in \pi} \bone_{\{\bx \in I_\bm\}} \ p_{\bm, \bh},  \quad \pi \in \Pi_{\delta}^{k}(\alpha,\beta),$$ 
is an $\eps$-net for $\Hold^{k}(\alpha,\beta)$ with $\eps = c_0 \delta$.
From Lemma \ref{lem:hold-count} and Theorem \ref{th:graph_cover}, we see that the system has entropy of order $O(\eps^{-k/\alpha})$, which is essentially the smallest possible \cite{MR0124720}.

\subsection{Approximating networks for H\"older immersions}
\label{sec:hold-im}

Define $\Hold^{k,d}(\alpha,\beta)$ as the class of functions $f:[0,1]^k \goto [0,1]^d$ with coordinates in $\Hold^k(\alpha,\beta)$, i.e. $f = (f_1,\dots,f_d)$ with $f_r \in
\Hold^k(\alpha,\beta)$ for $r = 1,\dots,d$.

For a function $f:[0,1]^k \to [0,1]^d$, define its image as
$$\im(f) = \{f(\bx): \bx \in [0,1]^k\}.$$
An approximation network for images of H\"older functions is simply built out of a tensor product of copies of the network built in the previous section.
Specifically, define the network $\bbG_\delta^{k,d}(\alpha,\beta)$, with nodes of the form $(\bm,\bh_1, \dots, \bh_d)$, indexing the multivariate polynomial $(p_{\bm,\bh_1}, \dots, p_{\bm,\bh_d})$; and edges between nodes indexing polynomial pieces on adjacent hypercubes and satisfying (\ref{eq:hold_neigh}) coordinate-wise.

\begin{example}{($k=1$, $\alpha \in (1,2]$)}
The polynomial pieces are of the form:
$$(h_1^{(0)} \delta + h_1^{(1)} \delta_1 (x - x_m), h_2^{(0)} \delta + h_2^{(1)} \delta_1 (x - x_m)), \quad |x - x_m| \leq \Delta/2.$$
\end{example}

\begin{lemma}
\label{lem:hold-count-im}
$\bbG_{\delta}^{k,d}(\alpha,\beta)$ has $O(\beta^{dc_3} \delta^{-dc_3} \Delta^{dc_4-k}) = O(\beta^{d c_3-(dc_4-k)/\alpha} \delta^{-dc_3 + (dc_4-k)/\alpha})$ nodes and each node has at most $2k 6^{dc_3}$ neighbors.
\end{lemma}
\begin{proof}
The proof follows that of Lemma \ref{lem:hold-count}.
Details are omitted.
\end{proof}

\begin{example}{($k=1$, $\alpha \in (1,2]$)}
$\bbG_{\delta}^{1,2}(\alpha,\beta)$ has $O(\beta^{2-1/\alpha}\delta^{-4+1/\alpha})$ nodes, each with degree at most 2592.
\end{example}

As before, to each vertex $(\bm,\bh_1, \dots, \bh_d)$ we associate a region around the image of $(p_{\bm,\bh_1}, \dots, p_{\bm,\bh_d})$ restricted to $I_\bm$, of thickness given by the error bound of Lemma  \ref{lem:local_approx}:
$$R (\bm,\bh_1, \dots, \bh_d) = \{\bz \in [0,1]^d: \exists \bx \in [0,1]^k, \forall r=1,\dots,d, \ |z_r - p_{\bm,\bh_r}(\bx)| \leq c_0 \delta \}.$$

Let $\Pi_{\delta}^{k,d}(\alpha,\beta)$ denote the set of connected components of $\bbG_{\delta}^{k,d}(\alpha,\beta)$, homeomorphic to the square grid $\{1, \dots, \Delta^{-1}\}^k$.

\begin{theorem}
\label{th:im_cover}
For each $f \in \Hold^{k,d}(\alpha,\beta)$, there is a connected component $\pi \in \Pi_{\delta}^{k,d}(\alpha,\beta)$ such that $\im(f) \subset R(\pi)$.
\end{theorem}
\begin{proof}
The proof follows that of Theorem \ref{th:graph_cover}.
Details are omitted.
\end{proof}

\subsection{Networks organized by scale}
\label{sec:hold-scale}

In practice, the parameters $\alpha$ and $\beta$ are often unknown, so that it becomes necessary and/or useful to look through a discrete set.
We introduce a scale parameter and another parameter indexing the approximation order, and organize the graphs accordingly.
Fix $J \in \bbN$ as the maximum scale and a sequence $a_J \to \infty$ as $J \to \infty$. 
At scale $j \in \{0,\dots,J\}$ and approximation order $\iota \in \bbN$, let $\bbG_{j,J}^{k}(\iota)$ (resp. $\bbG_{j,J}^{k,d}(\iota)$) be defined as $\bbG_\delta^{k}(\alpha, \beta)$ (resp. $\bbG_\delta^{k,d}(\alpha, \beta)$) with $\afloor = \iota$, $\Delta = 2^{-j}$, $\delta= 2^{j-J}$ and $|h^{(\bs)} \delta_{|\bs|}| \leq a_J^{|\bs|+1}$.

We have the following corollary of Theorems \ref{th:graph_cover} and \ref{th:im_cover}.

\begin{proposition}
\label{prop:discrete_cover}
For $\alpha, \beta > 0$, let $j = j(\alpha,\beta) = \lceil \frac{J + \log_2 (c_1\beta)}{1 + \alpha} \rceil$.
Assume $J$ is large enough that $j(\alpha,\beta) \leq J$ and $a_J \geq \beta$.
Also, take $\iota = \afloor$.
Then the covering result of Theorem \ref{th:graph_cover}  (resp. Theorem \ref{th:im_cover}) applies with $\bbG_{j,J}^{k}(\iota)$ (resp. $\bbG_{j,J}^{k,d}(\iota)$).
\end{proposition}
\begin{proof}
This is a simple corollary of Theorems \ref{th:graph_cover} and \ref{th:im_cover}.
\end{proof}

\subsubsection{A single multiscale network}
\label{sec:hold-multiscale} 

The networks $\bbG_{j,J}^{k}(\iota)$, $j = 1,\dots,J$, $\iota \in \bbN$, constitute a family of monoscale networks.
Instead, one may want to mix scales together (and possibly mix approximation orders too, though not done here) so as to better approximate functions with varying smoothness, akin to how Besov functions are decomposed into a sum of wavelets at various scales \cite{MR1228209}.
Specifically, consider $\cF^{k}(\iota)$ to be the class of $\iota$-times continuously differentiable functions of the form $f = \sum_{p \in P} f_p \chi_p$, where $P$ is a finite partition of $[0,1]^k$ into regions with boundaries of finite length and $f_p \in \cH^{k}(\alpha_p, \beta_p)$ with $\alpha_p \in (\iota, \iota+1]$ and $\beta_p > 0$.
For such a function class, we may want to involve a number of scales, each adapted to a different smoothness degree.

In particular, such a multiscale approximating network may be built out of the union of $\bbG_{j,J}^{k}(\iota)$, $j = 1,\dots,J$, with additional edges between nodes at different scales.
The neighboring condition across scales may be chosen to be identical to that defined in Section \ref{sec:hold-graph}, namely that $(\bm,\bh) \in \bbG_{j,J}^{k}$ and $(\bm_\star,\bh_\star)  \in \bbG_{j_\star, J}^{k}$ are neighbors if $I_{\bm}$ and $I_{\bm_\star}$ are adjacent, and if $p_{\bm, \bh}$ and $p_{\bm_\star, \bh_\star}$ together with their derivatives assume nearby values both at $\bx_\bm$ and $\bx_{\bm_\star}$.
(The condition translates into a precise statement involving $\bh$ and $\bh_\star$ akin to (\ref{eq:hold_neigh}), yet more cumbersome.  We omit details.)
Let $\bbG_J^k(\iota)$ denote this multiscale network.

Given a function $f \in \cF^{k}(\iota)$, $f = \sum_{p \in P} f_p \chi_p$, we use a recursive dyadic partitioning (RDP), the cornerstone of many multiscale algorithms \cite{MR1614527, 2D-beamlets, wedgelets, MR2128287}, to approximate the partition $P$.
Specifically, we start at $j=0$ and then recursively subdivide each dyadic hypercube $S$ at scale $j$ until $j \geq j(\alpha_p, \beta_p)$ for all $p \in P$ with $|p \cap S| > 0$.
We then use a polynomial piece from $\bbG_{j,J}^{k}(\iota)$ within each RDP cell at scale $j$.
See Figure \ref{fig:RDP-approx}.

For a subset of nodes $\pi$ in $\bbG_J^{k}(\iota)$, define $R(\pi) = \bigcup_{(\bm,\bh) \in \pi} R(\bm,\bh)$, where each region is defined with the appropriate scale.

\begin{proposition}
\label{prop:RDP-approx}
For each function $f \in \cF^{k}(\iota)$, there is a connected component $\pi$ within $\bbG_J^{k}(\iota)$ in correspondence with an RDP such that $\graph(f) \subset R(\pi)$.
\end{proposition}
\begin{proof}
This is essentially a corollary of Proposition \ref{prop:discrete_cover}.
Details are omitted.
\end{proof}

Note however that the typical degree of a node in $\bbG_J^{k}(\iota)$ increases with $J$, as a result of connecting nodes across scales.

\begin{figure}[htbp]
  \centering
    \includegraphics[height = 2.5in]{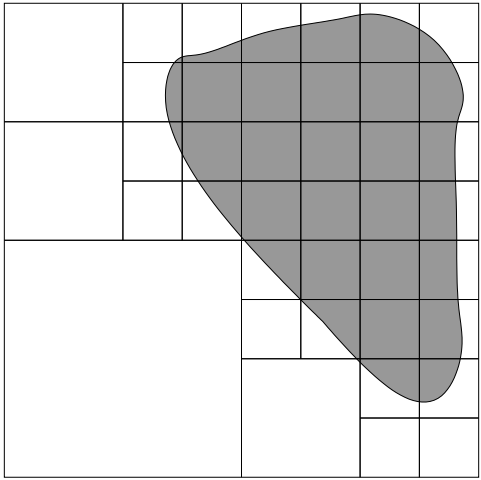}
  \caption{A partition into two regions $p_1$ (white) and $p_2$ (grey) associated with some function $f = f_{p_1} \chi_{p_1} + f_{p_2} \chi_{p_2}$, with corresponding scales $j(\alpha_{p_1}, \beta_{p_1}) = 1$ and $j(\alpha_{p_2}, \beta_{p_2}) = 3$.}
  \label{fig:RDP-approx}
\end{figure}

\section{Detection of Objects in Point Clouds and Images}
\label{sec:applications}

We consider a simple model of detection of objects in point clouds and design some algorithms based on the networks introduced in Section \ref{sec:hold}.
The objects are assumed to be of H\"older type.
Fix an H\"older class $\Hold^{k,d-k}(\alpha, \beta)$.
For $f \in \Hold^{k,d-k}(\alpha, \beta)$ and $\eta \geq 0$, define
$$\graph_\eta(f) = \{(\bx,\bz) \in [0,1]^d: \|\bz - f(\bx)\| \leq \eta\},$$
which is a region centered around the graph of $f$ and of thickness $2 \eta$.

Suppose we observe a point cloud $X_1, \dots, X_n \in [0,1]^d$, and want to decide between the following two hypotheses (generative models):
$$\begin{array}{rll}
H_0: &\quad X_1, \dots, X_n \sim^{\rm iid} &\Uniform[0,1]^d; \nonumber \\
H_1: &\quad X_1, \dots, X_n \sim^{\rm iid} & (1-\eps_n)\Uniform[0,1]^d + \eps_n \Uniform(\graph_\eta(f^*)),  \nonumber \\
& & \text{for some (unknown)} \ f^* \in \Hold^{k,d-k}(\alpha, \beta).
\end{array}$$
The same situation was considered in \cite{MSDFS, ery-thesis, CTD}.

For a (measurable) set $S \subset [0,1]^d$, let $N(S)$ denote the number of data points belonging to $S$, which under the null has the binomial distribution with parameters $n$ and $|S|_d$, the $d$-dimensional Lebesgue measure of $S$:
$$N(S) = \# \{i: X_i \in S\} \sim^{H_0} \Bin(n, |S|_d).$$

\subsection{Generalized Likelihood Ratio Test}

If $f^*$ were known, the most powerful test would be the likelihood ratio (i.e. Neyman-Pearson) test, which rejects for large values of $N(\graph_\eta(f^*))$.
The scan statistic is the maximum over those statistics:
$$M_\eta^{k,d}(\alpha, \beta) = \max_{f \in \Hold^{k,d}(\alpha, \beta)} N(\graph_\eta(f)).$$
The generalized likelihood ratio test (GLRT) rejects when $M_\eta^{k,d}(\alpha, \beta)$ is large.

Define 
$$\rho(k,d,\alpha) = \frac{k}{k + \alpha(d-k)}.$$

\begin{theorem}
\label{th:GLRT}
There are constants $A, B > 0$ not depending on $n$ such that
$$\pr{A (n^{\rho} \vee \eta^{d-k} n) \leq M_\eta^{k,d}(\alpha, \beta) \leq B (n^{\rho} \vee \eta^{d-k} n) | H_0} \to 1, \quad n \to \infty.$$ 
\end{theorem}
\begin{proof}
The lower bound is obtained by interpolation of carefully selected points, while the upper bound is obtained using a precise-enough net for $\Hold^{k,d-k}(\alpha, \beta)$ of  near-optimal entropy (such as introduced in Section \ref{sec:hold}).
We refer the reader to Section 2.3 in \cite{ery-thesis} for more details.
\end{proof}

Since $M_\eta^{k,d}(\alpha, \beta) \geq N(\graph_\eta(f^*))$ and $\graph_\eta(f^*)$ contains at least an $\eps_n$ proportion of the point cloud (roughly), the GLRT asymptotically separates $H_0$ and $H_1$ if, for some fixed $B' > B$, $\eps_n \geq B' (n^{\rho} \vee \eta^{d-k} n)$, meaning that both the probabilities of false alarm (type I error) and missed detection (type II error) tend to zero as $n$ increases.

The GLRT as defined above is challenging, if not impossible to compute exactly.
We use instead an approximation based on the coverings constructed in Section \ref{sec:hold}, turning an optimization over a functional space into an optimization over paths in a network, for which a large number of algorithms have been developed.
See also \cite{kirsanov2004dgm,1238310}, where variational problems related to computing minimal surfaces are turned into combinatorial optimizations over paths and other structures within networks.

Recall the notation used in Section 2, and assume again that $\Delta$ and $\delta$ are related according to (\ref{eq:Delta-delta}).
We use the network $\bbG_\delta^{k,d-k}(\alpha,\beta)$ to build appropriate coverings, this time with slightly enlarged regions:
$$\oR(\bm, \bh_1, \dots, \bh_{d-k}) = \{(\bx,\bz): \bx \in I_\bm, \forall r=1, \dots, d-k, \ |z_r - p_{\bm,\bh_r}(\bx)| \leq (c_0 + 1) \delta\},$$
where $c_0 = (2^{-\alpha} + c_2/2)$ as in (\ref{eq:R}).
 
Fix a path $\{\bm_{\rm zz}(t): t=1, \dots, \Delta^{-k}\}$ in the square grid $\{1, \dots, \Delta^{-1}\}^k$ covering the whole grid in zig-zag fashion, and define $\cP_{\delta}^{k,d-k}(\alpha, \beta)$ as the set of paths in $\bbG_\delta^{k,d-k}(\alpha,\beta)$ of the form 
$$\{(\bm_{\rm zz}(t),\bh_1(t), \dots, \bh_{d-k}(t)): t=1, \dots, \Delta^{-k}\}.$$
Also, recall the definition of $\Pi_\delta^{k,d-k}(\alpha, \beta)$ in Section \ref{sec:hold}.
Now consider the following alternative statistics:
$$M_{\Pi,\delta}^{k,d}(\alpha, \beta) = \max_{\pi \in \Pi_{\delta}^{k,d-k}(\alpha,\beta)} N(\oR(\pi)), \qquad M_{\cP,\delta}^{k,d}(\alpha, \beta) = \max_{P \in \cP_\delta^{k,d-k}(\alpha, \beta)} N(\oR(P)).$$
In that case, 
$$M_\eta^{k,d}(\alpha, \beta) \leq M_{\Pi,\delta}^{k,d}(\alpha, \beta) \leq M_{\cP,\delta}^{k,d}(\alpha, \beta).$$
The first inequality comes from Theorem \ref{th:graph_cover} and the triangle inequality; the second from the fact that $\Pi_\delta^{k,d-k}(\alpha, \beta) \subset \cP_\delta^{k,d-k}(\alpha, \beta)$.
Actually, by Theorem \ref{th:GLRT-approx} below, with a proper choice for $\delta$ all three are of same order of magnitude with high probability.

\begin{theorem}
\label{th:GLRT-approx}
With $\delta = n^{-\alpha/(k+\alpha(d-k))} \vee \eta$, there is a constant $C = C(k,d,\alpha,\beta)$ such that
$$\pr{M_{\cP,\delta}^{k,d}(\alpha, \beta) \leq C (n^{\rho} \vee \eta^{d-k} n)| H_0} \to 1, \quad n \to \infty.$$
\end{theorem}
\begin{proof}
The proof of Theorem \ref{th:GLRT-approx} is in the Appendix.
\end{proof}


Computing $M_{\cP,\delta}^{k,d}(\alpha, \beta)$ may be done efficiently using dynamic programming ideas, for example as implemented in \cite{MR2379113}; this is due to the fact that the paths in $\cP_\delta^{k,d-k}(\alpha, \beta)$ are oriented and with no loops.
Though $M_{\Pi,\delta}^{k,d}(\alpha, \beta)$ provides a better approximation to the scan statistic, we do not know of an efficient way to compute it directly.

The results above hold for H\"older immersions as well, though in that case the computations are much more challenging.
This comes from the fact that H\"older immersions may self-intersect, so that dynamic programming approaches do not apply.
In fact, if the optimization is over all paths of length $\Delta^{-k}$ instead, the setting is equivalent to the Budget-Reward Problem \cite{DasHesSon} (also called Bank Robber Problem), closely related to the Prize-Collecting Traveling Salesman problem \cite{prize-collecting}; though those problems are known to be NP-hard, there are polynomial-time approximations \cite{DasHesSon}.
Other approaches have been suggested in this situation, for example in \cite{2D-beamlets}, where ratios of additive criteria are used to recover chains of beamlets; or algorithms implemented to extract curves from saliency networks \cite{590008, 649458}.

\subsection{Longest Significant Run}

We propose an alternative approach based on the size of the longest path after discarding nodes with low counts within their associated region, which in effect generalizes the algorithm introduced in \cite{MSDFS}.

For a threshold $\tau > 0$, define 
$$S(\bm, \bh_1, \dots, \bh_{d-k}) = 
\left\{\begin{array}{ll}
1, & N(\oR(\bm,\bh_1, \dots, \bh_{d-k})) > \tau n \delta^{k/\alpha + d-k},\\
0, & \text{otherwise}.
\end{array}\right.
$$
Then define $L_\tau^{k,d}(\alpha, \beta)$ as the length of the longest path of the form 
$$\{(\bm_{\rm zz}(t), \bh_1(t), \dots, \bh_{d-k}(t)): t = t_0, \dots, t_0 + \ell -1\},$$ 
such that $S(\bm_{\rm zz}(t), \bh_1(t), \dots, \bh_{d-k}(t)) = 1$ for all $t = t_0, \dots, t_0 + \ell -1$.
If we see each $S(\bm, \bh_1, \dots, \bh_{d-k})$ as a test at node $(\bm, \bh_1, \dots, \bh_{d-k})$ and say that it is significant if it equals 1, then $L_\tau^{k,d}(\alpha, \beta)$ is the length of the longest significant run (LSR).

\begin{theorem}
\label{th:LSRT}
With $\delta = n^{-\alpha/(k+\alpha(d-k))} \vee \eta$, for $\tau$ large enough,
$$\pr{L_\tau^{k,d}(\alpha, \beta) \leq \log(1/\delta)|H_0} \to 1, \quad n \to \infty.$$
Also, there is a constant $C = C(k,d,\alpha,\beta)$ such that, if $\eps > C n^\rho \vee \eta^{d-k}n$,
$$\pr{L_\tau^{k,d}(\alpha, \beta) > \log(1/\delta)|H_1} \to 1, \quad n \to \infty.$$
\end{theorem}
\begin{proof}
The proof of Theorem \ref{th:LSRT} is in the Appendix.
\end{proof}

Therefore, the LSRT achieves the detection rate established for the GLRT (Theorem \ref{th:GLRT}) and its approximation (Theorem \ref{th:GLRT-approx}). 
Moreover, it can be computed using dynamic programming ideas since again the paths considered are oriented and without loops.
This is done in \cite{MSDFS}.

The same approach applies essentially unchanged to the case of H\"older immersions, though with pathological cases obscuring the exposition, so that we omit details.

A similar approach is advocated in \cite{water-quality}, where a network of streams is monitored for pollution levels and large connected components of areas marked as problematic (polluted) are of interest. 
Two of the same researchers suggest a hybrid method in \cite{patil-upper} applied to the identification of regions of interest (hot spots) in raster maps.

\subsection{Detection in Grey-Level Images}

The results we developed for point clouds can be obtained (in slightly different form) for digitized images, which is the context for the experiments of Section \ref{sec:experiments}.

Suppose we observe a $d$-dimensional pixel array $Y$, with a total of $n$ pixels, of the form:
$$Y = \mu \ \xi_{\graph_\eta(f^*)} + \sigma Z,$$
where $Z$ is white Gaussian noise, with independent, standard normal entries; $\sigma > 0$ is the noise level; $\mu$ is the signal level; and for a subset $S \subset [0,1]^d$, $\xi_S$ is the array with $\ell^2$-norm 1 identifying the pixels that $S$ intersects, namely $\xi_S(\bi) \propto \bone\{S \cap \pix(\bi) \neq \emptyset\}$.

We observe $Y$ and want to decide between $H_0$ and $H_1$ below:
$$\begin{array}{rll}
H_0: &\quad \mu = 0; \\
H_1: &\quad \mu > 0, \text{ and } f^* \in \Hold^{k,d-k}(\alpha, \beta) \text{ is unknown}.
\end{array}$$
This setting is considered in \cite{MGD} (with a slightly different definition for $\xi_S$) in the context of parametric objects.

Following the same arguments as for point clouds, we find that the GLRT asymptotically separates $H_0$ and $H_1$ if $\mu \geq C (n^{k/(2 \alpha d)} \vee \eta^{-k/(2 \alpha)})$ with $C$ large enough, and that a detection threshold of same order of magnitude is achieved both by the approximate GLRT and the LSRT, with $\delta = \eta \vee n^{-1/d}$.
We omit details.

\section{Beamlet and Beamlet-Like Networks}
\label{sec:curves}

We now focus on curves ($k = 1$) of H\"older smoothness with $\alpha \in (1,2]$.
The corresponding net described in Section \ref{sec:hold} is made of piecewise linear functions.
In this section, we show that such a net may be obtained by carefully chaining beamlets as suggested in \cite{2D-beamlets}, yielding a more economical net at a comparable degree of approximation.

The case of curves is special in the sense that it is the simplest, and in particular allows us to use a parametrization by arclength.
Higher dimensional surfaces, even though smooth, may exhibit strange behavior, for example a very thin 2D surface in 3D.

\subsection{H\"older Curves}

We adopt here a slightly more intrinsic definition for curves.
For $\alpha \in (1,2]$ and $\lambda, \kappa > 0$, let
$\Gamma(\alpha,\lambda,\kappa)$ be the set of curves $\gamma \subset [0,1]^d$
with $\length(\gamma) \leq \lambda$ and parametrized by arclength with
\begin{equation}
\label{eq:beam_taylor1}
\|\gamma(t)-\gamma(s) - (t-s) \gamma'(s)\| \leq \kappa\
|t-s|^\alpha,
\quad \forall s,t \in [0,\length(\gamma)].
\end{equation}
$\Gamma(\alpha,\lambda,\kappa)$ is in close correspondence with $\Hold^{1,d}(\alpha,\beta)$ as defined in Section 2.
Note that the case $\alpha=2$ includes all twice differentiable curves with curvature bounded by $2 \kappa$.

Curves in $\Gamma(\alpha,\lambda,\kappa)$ satisfy the following properties.

\begin{lemma} \label{lem:beam_taylor_derivative}
For all $\gamma \in \Gamma(\alpha,\lambda,\kappa)$,
$$\|\gamma'(t) - \gamma'(s)\| \leq 2 \kappa |t-s|^{\alpha-1}, \quad \forall s,t.$$
\end{lemma}
\begin{proof}
Fix $0 \leq s < t \leq \length(\gamma)$.
The triangle inequality and (\ref{eq:beam_taylor1}) give
\begin{eqnarray*}
(t-s) \|\gamma'(t) - \gamma'(s)\| & \leq & \|\gamma(t) - \gamma(s) -
\gamma'(s)(t-s)\| + \|\gamma(s) - \gamma(t) - \gamma'(t)(s-t)\|\\
& \leq & 2 \kappa (t-s)^\alpha.
\end{eqnarray*}
\end{proof}

\begin{lemma} 
\label{lem:beam_taylor2}
Let $\gamma \in \Gamma(\alpha,\lambda,\kappa)$. 
For all arclengths $r < s < t$,
$$\left\|\gamma(s) - \gamma(r) - \frac{s - r}{t - r}\ (\gamma(t)
- \gamma(r))\right\| \leq 2 \kappa (t - r)^\alpha.$$
\end{lemma}
\begin{proof}
Applying (\ref{eq:beam_taylor1}) twice yields
$$\|\gamma(s) - \gamma(r) - (s-r) \gamma'(r)\| \leq \kappa (s-r)^\alpha \leq
\kappa (t-r)^\alpha$$
and 
$$\left\|\frac{s - r}{t - r}\ (\gamma(t) - \gamma(r)) - (s-r) \gamma'(r)\right\|
\leq \kappa\ (s-r) (t-r)^{\alpha-1} \leq \kappa (t-r)^\alpha.
$$
Then apply the triangle inequality and conclude.
\end{proof}

\subsection{Beamlets}
\label{sec:beamlets}

Beamlets were introduced in 2D by Donoho and Huo \cite{2D-beamlets}, and then in 3D by Donoho and Levi \cite{3D-beamlets, spie-beamlets}.
We define them here in any dimension $d \geq 2$.
Fix a maximum scale $J \in \bbN$.
Define $\delta_0 = 2^{-J}$ and at scale $j \in \{0,\dots,J\}$, define $\Delta = 2^{-j}$.
For a given coordinate $r = 1, \dots, d$, hyperplanes of the form 
$$\{(x_1, \dots, x_d): x_r = h \Delta\},$$
where $h \in \{0,\dots,\Delta^{-1}\}$, are called $r$-hyperplanes.
Such hyperplanes will be called $\Delta$-hyperplanes; they partition the unit hypercube $[0,1]^d$ into smaller hypercubes of sidelength $\Delta$ that we call $\Delta$-hypercubes.
On each $\Delta$-hyperplane, we consider a regular square grid with spacing $\delta_0$.
Formally, we consider gridpoints of the form $(h_1 \delta_0, \dots,h_d \delta_0)$, $h_r = 0, 1, \dots, \delta_0^{-1}$, with at least one coordinate an integer multiple of $\Delta \delta_0^{-1}$; 
if this happens at the $r$th coordinate, we speak of an $r$-gridpoint, which by definition belongs to an $r$-hyperplane.
A beamlet is simply a line-segment joining two gridpoints belonging to the same $\Delta$-hypercube.
See Figures \ref{fig:2D-beamlets}. 

\begin{figure}[htbp]
  \centering
  \subfloat[Coarsest scale, j=0]{%
    \includegraphics[height = 2in]{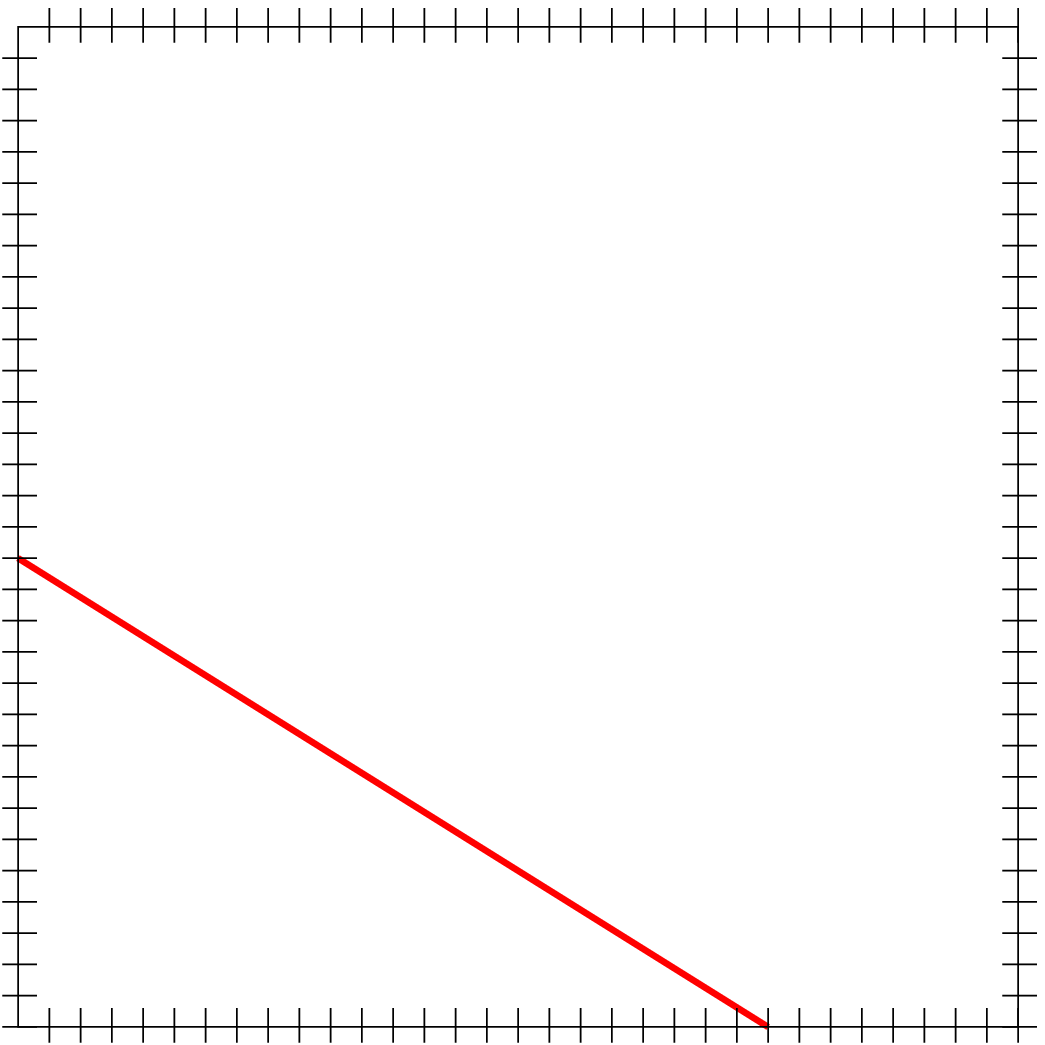}}%
  \quad%
  \subfloat[Next finer scale, j=1]{%
    \includegraphics[height = 2in]{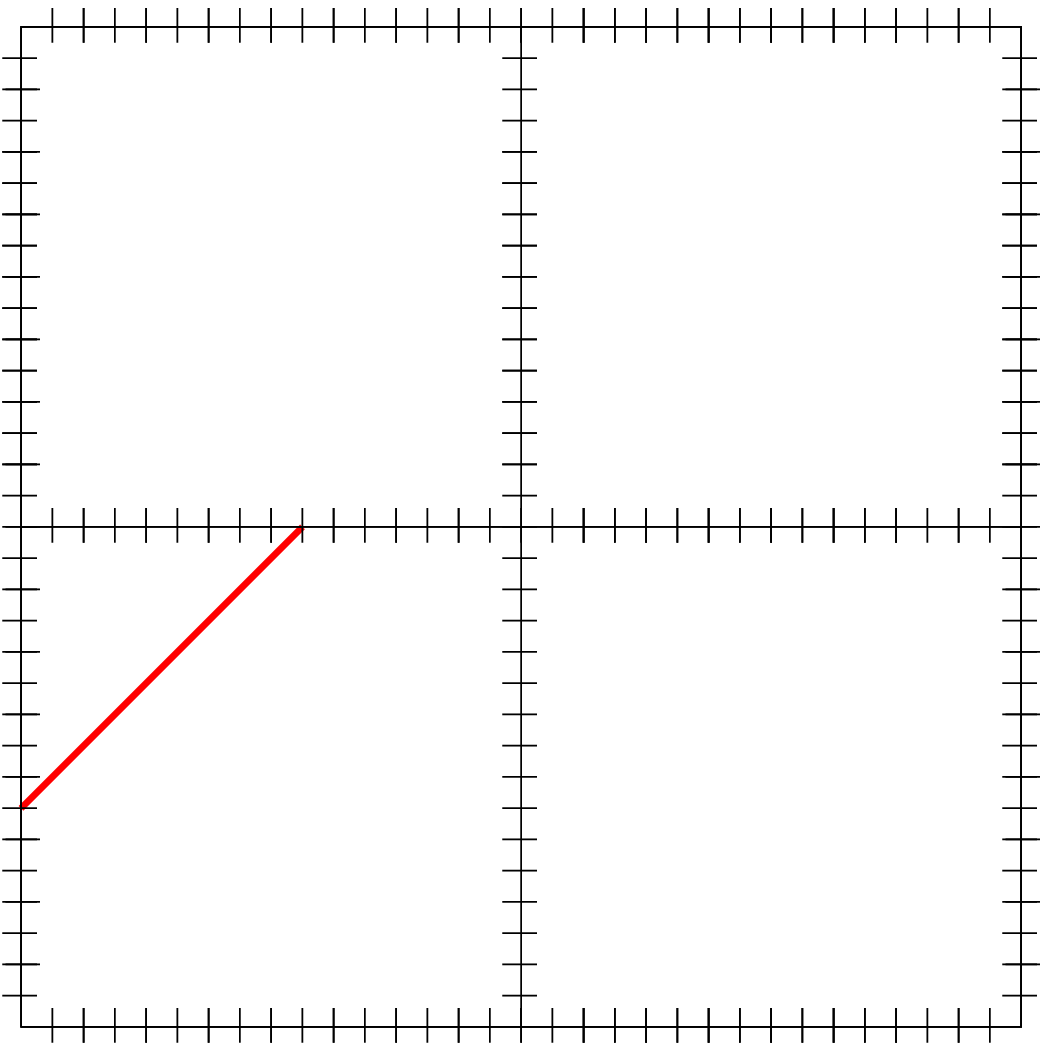}} \\
  \subfloat[Coarsest scale, j=0]{%
    \includegraphics[trim = 10mm 40mm 135mm 30mm, clip=true,height = 2in]{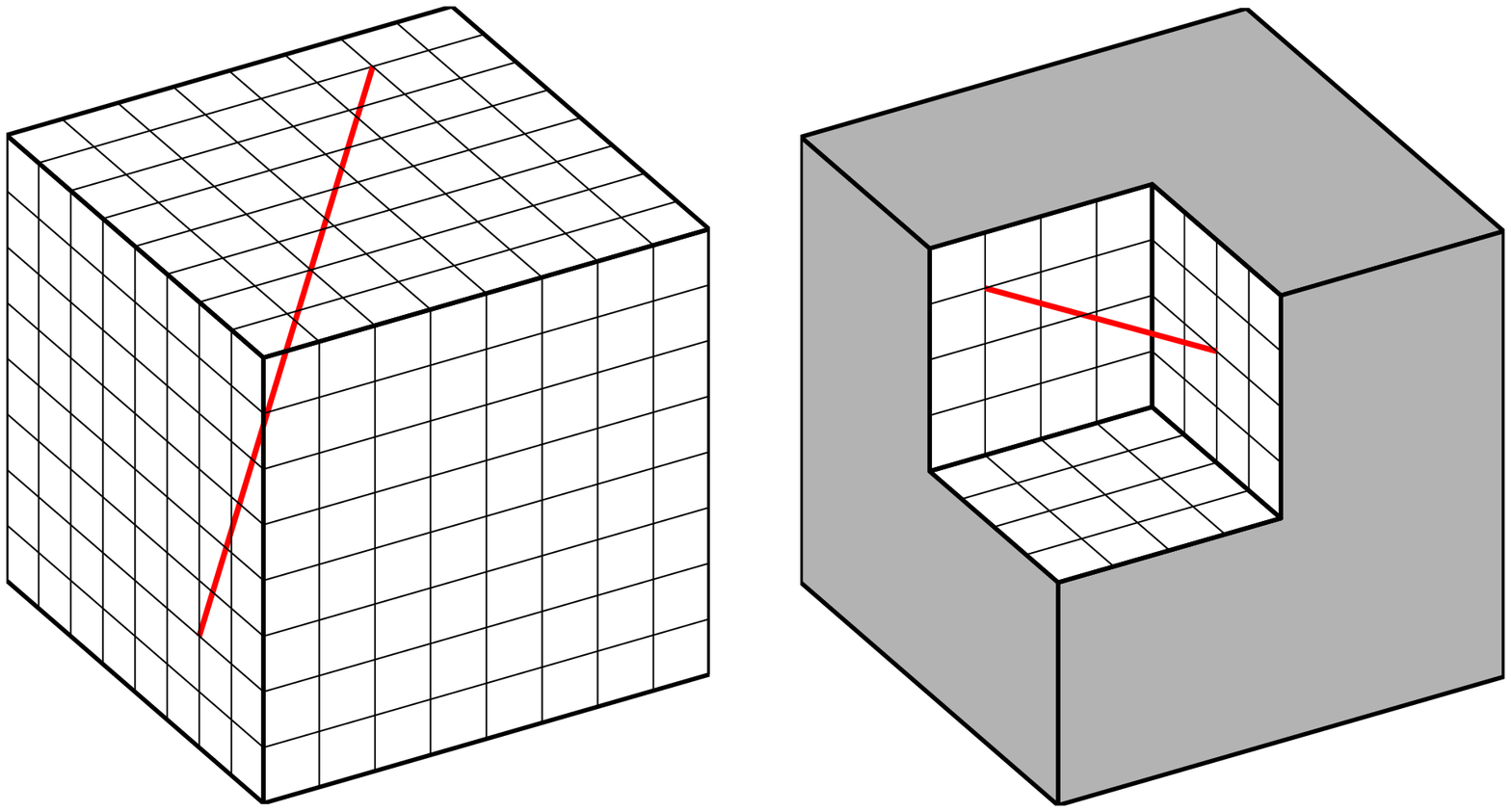}}%
  \quad%
  \subfloat[Next finer scale, j=1]{%
    \includegraphics[trim = 130mm 40mm 35mm 30mm, clip=true,height = 2in]{figs//fig_beamlets.eps}}%
    \caption{Examples of beamlets in 2D and 3D.}
  \label{fig:2D-beamlets}
\end{figure} 


The beamlet graph in \cite{2D-beamlets} refers to the graph with nodes the gridpoints and edges the beamlets.
Our interest instead is in building a beamlet good-continuation graph (i.e. network), with nodes the beamlets themselves and edges between beamlets in good-continuation, such that paths in that network provide a useful net for smooth curves.   
Such a graph is used in some experiments performed in \cite{2D-beamlets}, yet never formally defined there or elsewhere, though a related construction is presented in \cite{MSDFS} for curves that are graphs of H\"older functions.
Formally, two beamlets are neighbors if they are of the form
$[b_1,b_2]$ and $[b_2,b_3]$ and satisfy
\begin{eqnarray}
&& \|(b_{r,2} - b_{r,1})(b_{3} - b_{2}) - (b_{r,3} - b_{r,2})(b_{2} - b_{1})\| \nonumber \\
&& \qquad \qquad \qquad \qquad \qquad \qquad
 \leq 2^{j-J} (\|b_{3} - b_{2}\|  + \|b_{2} - b_{1}\|), \quad \forall r = 1, \dots, d. \label{eq:beamlet-neigh}
\end{eqnarray}
This is basically a constraint on the angle formed by $[b_1,b_2]$ and $[b_2,b_3]$; see Figure \ref{fig:2D-beamlets-neigh}.

\begin{figure}[htbp]
  \centering
\includegraphics[height = 2in]{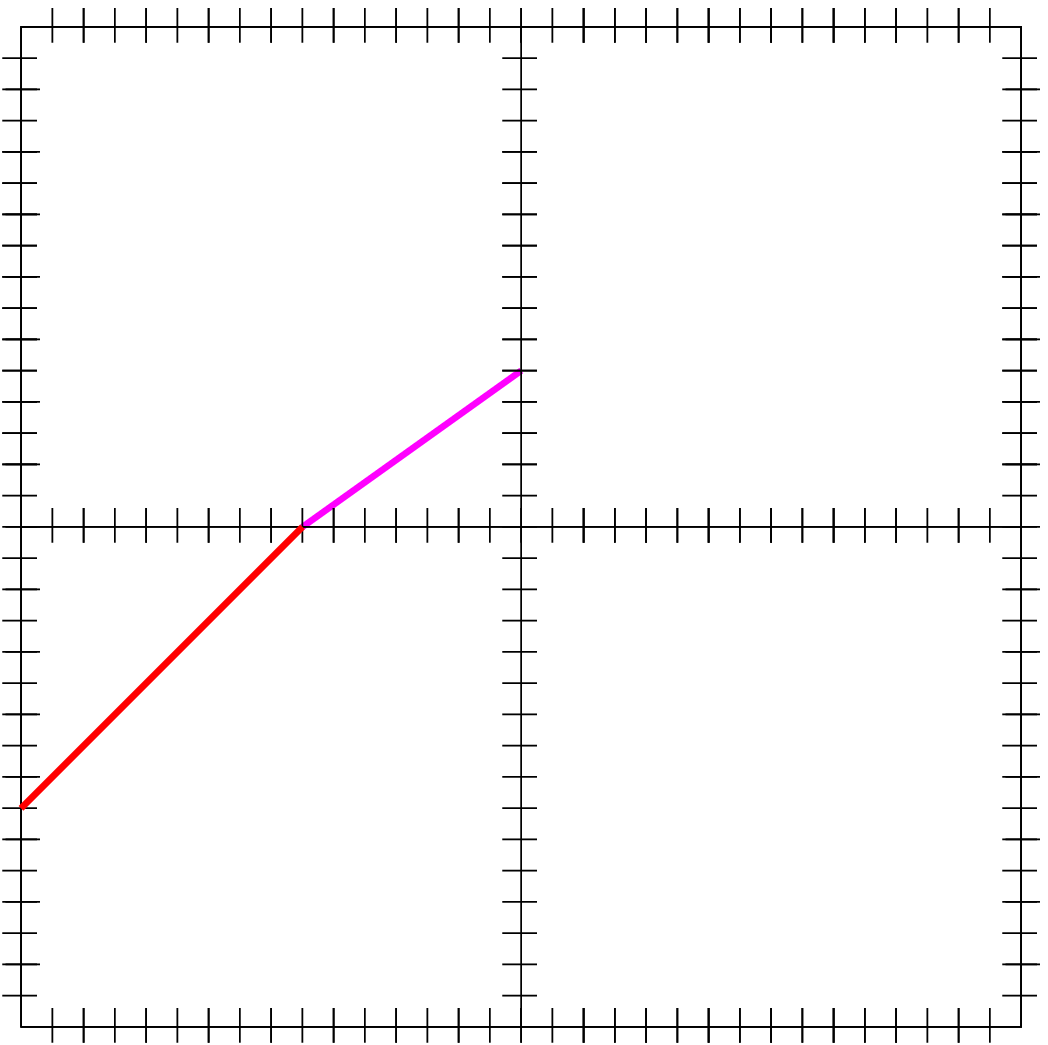}
  \qquad%
    \includegraphics[height = 2in]{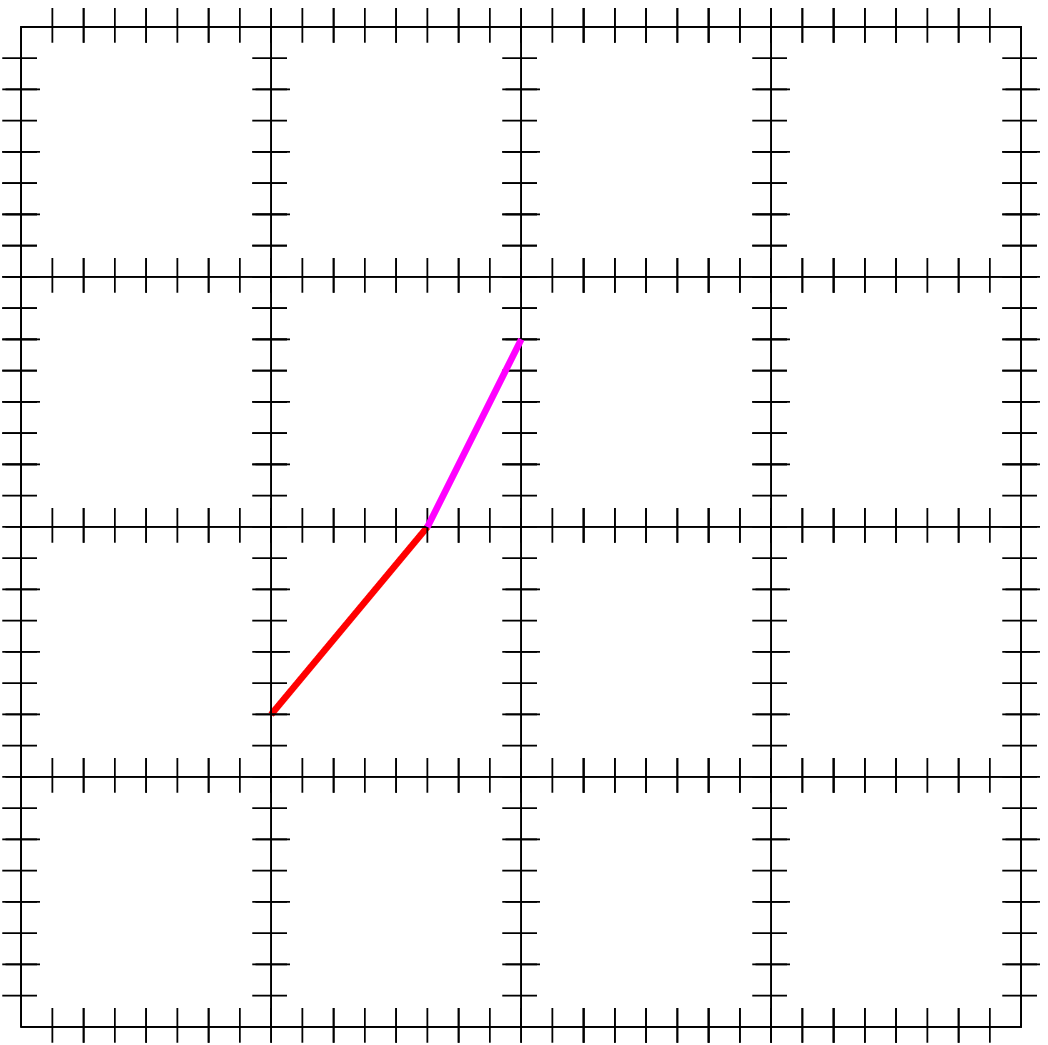}
  \caption{Examples of 2D beamlets in good-continuation.}
  \label{fig:2D-beamlets-neigh}
\end{figure}

Let $\bbB_{j,J}^d$ denote the resulting beamlet good-continuation network at scale $j$.   

\begin{lemma}
\label{lem:beamlet-count}
The number of beamlets at scale $j$ is of order $\bbB_{j,J}^d = O(d^2 2^{2(d-1)J - (d-2)j})$ and a given beamlet $B \in \bbB_{j,J}^d$ of length $|B|$ has $O(d |B|^{-(d-1)})$ neighbors. 
\end{lemma}
\begin{proof}
The proof of Lemma \ref{lem:beamlet-count} is in the Appendix.
\end{proof}

\begin{remark}
Short beamlets in the network $\bbB_{j,J}^d$ are highly connected.
This is an undesirable feature and we were not able to avoid it, other than defining a closely related system in Section \ref{sec:beams}.
\end{remark}

To each beamlet $B$ we associate a tubular region:
$$R(B) = \{\bx \in [0,1]^d: \min_{\by \in B} \|\bx - \by\| \leq 2^{j-J}\}.$$
We extend this definition to subsets of beamlets $R(\pi) = \bigcup_{B \in \pi} R(B)$.

\begin{theorem}
\label{th:beamlet-curve-covering}
Fix $\alpha \in (1,2]$ and $\lambda,\kappa>0$.
Let $j(\alpha, \kappa) = \left\lceil \frac{\log_2(\kappa) + J}{\alpha + 1} \right\rceil.$
There is a universal constant $K$ such that, when $J$ is large enough and $j \geq j(\alpha,\beta) + K$, to each curve $\gamma \in \Gamma(\alpha,\lambda,\kappa)$ corresponds a path $\pi_{j}$ in $\bbB_{j,J}^d$ chaining at most $(2d)(\lambda 2^j + 2)$ beamlets such that $\gamma$ is included in $R(\pi_j)$.
\end{theorem}
\begin{proof}
Theorem \ref{th:beamlet-curve-covering} is a consequence of Theorem \ref{th:curve-covering}; see the comments following Theorem \ref{th:curve-covering}.
\end{proof}

Such a covering is illustrated in Figure \ref{fig:2D-beamlets-approx}.

\begin{figure}[htbp]
  \centering
\includegraphics[height = 2.5in]{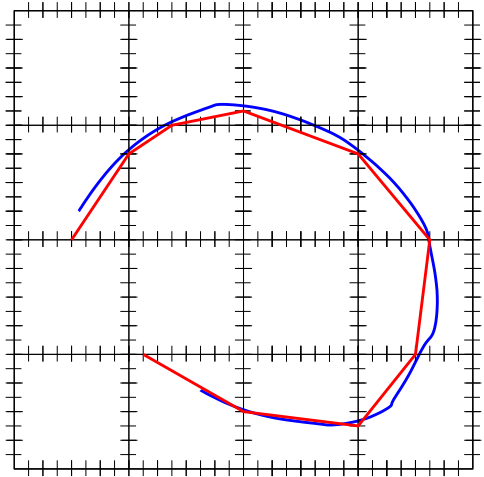}
  \caption{Example of a curve (blue) approximation by a chain of 2D beamlets (red) in good continuation.}
  \label{fig:2D-beamlets-approx}
\end{figure}

\subsubsection{Multiscale beamlet good-continuation network}

The chaining algorithms proposed in \cite{2D-beamlets} involve linking beamlets at different scales.
Such a multiscale network may simply be constructed by taking the union of all $\bbB_{j,J}^d, j=0,\dots,J$ and adding edges between beamlets at different scale satisfying (\ref{eq:beamlet-neigh}) with $j$ replaced by the maximum of the two scales involved; see Figure \ref{fig:2D-beamlets-neigh-two-scales}.
As in Section \ref{sec:hold-multiscale}, an approximation is built using an RDP, this time of the ambient space $[0,1]^d$.
Such a network provides more precise approximations to curves with varying smoothness, however at the cost of (substantially) increasing connectivity. 

\begin{figure}[htbp]
  \centering
\includegraphics[height = 2in]{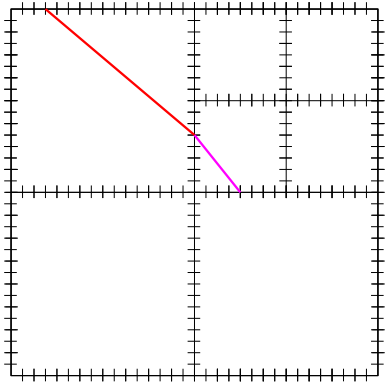}
  \caption{Example of 2D beamlets at different scales in good continuation.}
  \label{fig:2D-beamlets-neigh-two-scales}
\end{figure}

\subsection{Beams}
\label{sec:beams}

We found in Lemma \ref{lem:beamlet-count} that defining a notion of good continuation directly between beamlets is problematic in that small corner beamlets become hubs, connected to a large number of other beamlets.
The construction presented here uses other line-segments akin to beamlets, that we call beams.
(Note that the term `beam' refers to something else in \cite{2D-beamlets}.)
Beams at a given scale are of similar length, thus avoiding the problem just mentioned, and each beam can be well approximated by a short chain of beamlets at that same scale.

Another drawback of beamlets as defined in Section \ref{sec:beamlets} is that the quantization by $\delta_0$ does not change with the scale as $\Delta$ does; this results in a system that is more rich than needed, therefore wasteful.
Just as in Section \ref{sec:hold-scale}, at scale $j$ we define $\delta = 2^{j-J}$, and $\Delta = 2^{-j}$ as before; we assume that $j \leq J/2$, so that $\Delta$ is an integer multiple of $\delta$.

We define two kinds of beams.
For $r = 1, \dots, d$, an $r$-beam is a line segment joining two $r$-gridpoints $b_1$ and $b_2$ belonging to the same $\Delta$-hypercube on opposite sides, and such that 
\begin{equation} \label{eq:h-beam}
\|b_1 - b_2\| \leq \Delta + \delta.  
\end{equation}
Therefore, an $r$-beam makes an angle of about 45 degrees or less with the $r$th axis.
 
For $r_1, r_2 = 1, \dots, d$, with $r_1 \neq r_2$, an $r_1r_2$-beam is a line segment joining an $r_1$-gridpoint $b_1$ and a $r_2$-gridpoint $b_2$ such that
\begin{align}
\Delta \leq |b_{r_1,1} - b_{r_1,2}| \vee |b_{r_2,1} - b_{r_2,2}| &< 2 \Delta, \label{eq:hh-beam-1}\\
\left|\ |b_{r_1,1} - b_{r_1,2}| - |b_{r_2,1} - b_{r_2,2}|\ \right| &\leq \delta, \label{eq:hh-beam-2}\\ 
(|b_{r_1,1} - b_{r_1,2}| \vee |b_{r_2,1} - b_{r_2,2}|) - |b_{r_3,1} - b_{r_3,2}| &\geq - \delta, \quad \forall r_3 = 1, \dots, d \label{eq:hh-beam-3}.
\end{align}
Note that $b_1$ and $b_2$ do not belong to the same $\Delta$-hypercube and that (\ref{eq:hh-beam-3}) is void in dimension $d=2$.
Hence, an $r_1r_2$-beam connects an $r_1$-hyperplane and a $r_2$-hyperplane, making angles of about 45 degrees at the intersection with those hyperplanes.  
The reason $r_1r_2$-beams are so restricted is that they are only used to connect $r_1$-beams and $r_2$-beams.

\begin{figure}[htbp]
  \centering
  \subfloat[$r2$-beam ($j=0$)]{%
    \includegraphics[height = 2in]{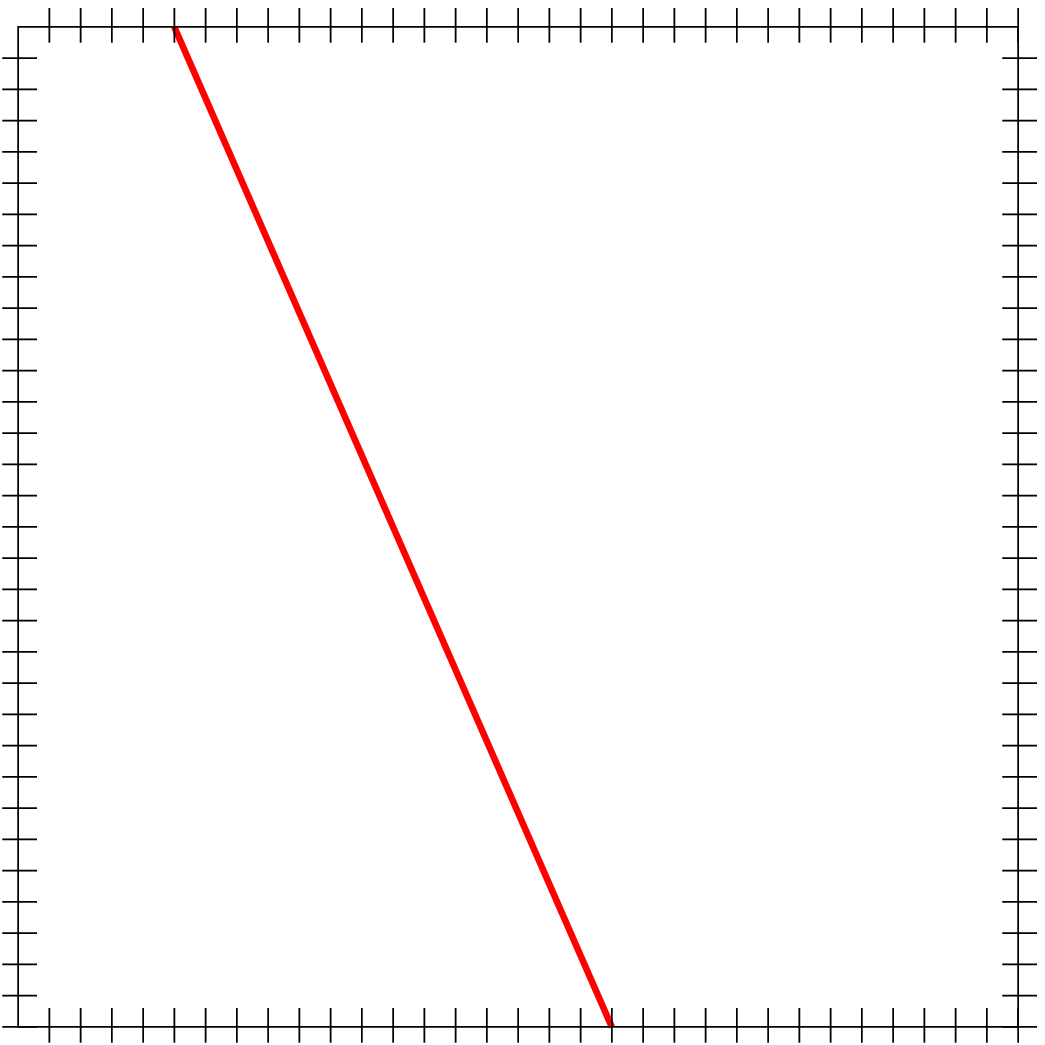}}%
  \quad%
  \subfloat[$r1$-beam ($j=1$)]{%
    \includegraphics[height = 2in]{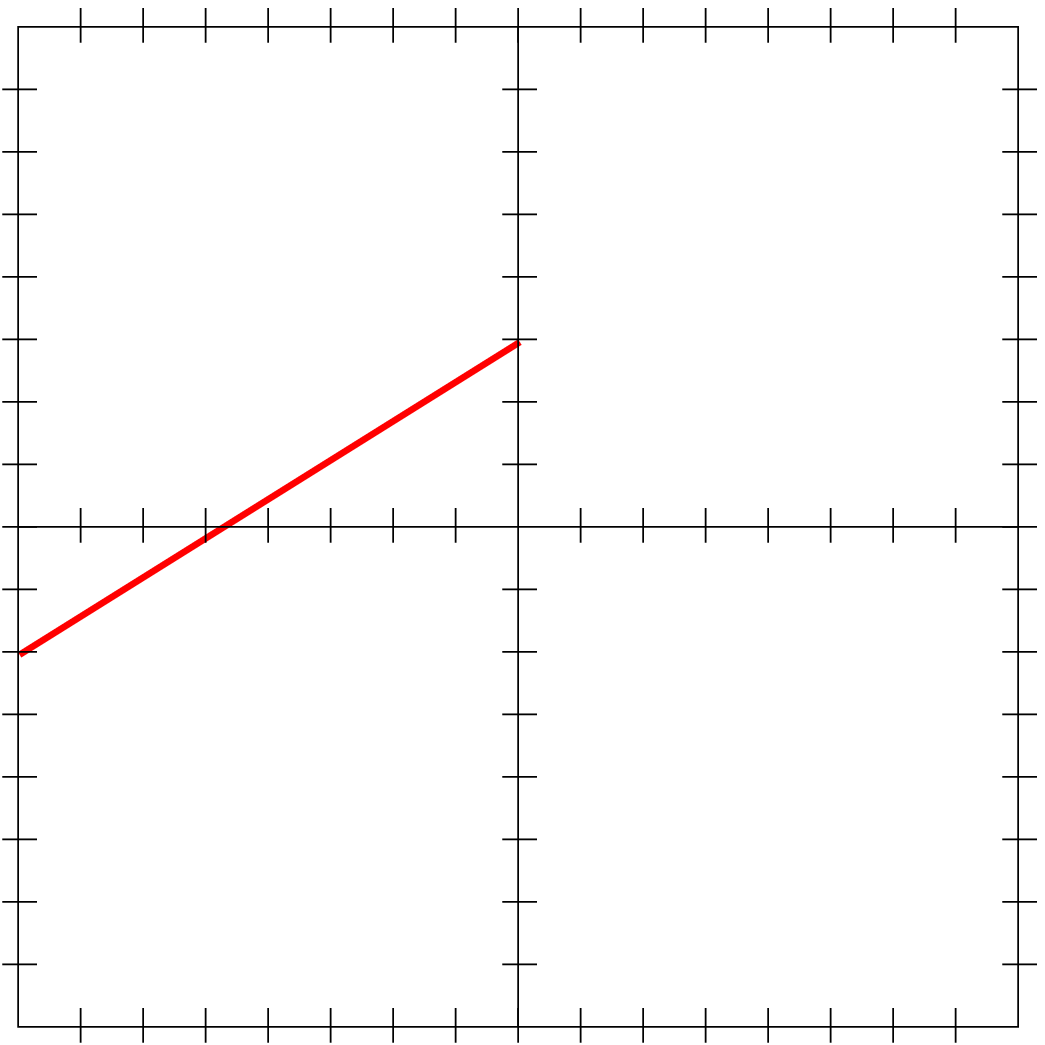}}%
  \quad%
  \subfloat[$r1r2$-beam ($j=2$)]{%
    \includegraphics[height = 2in]{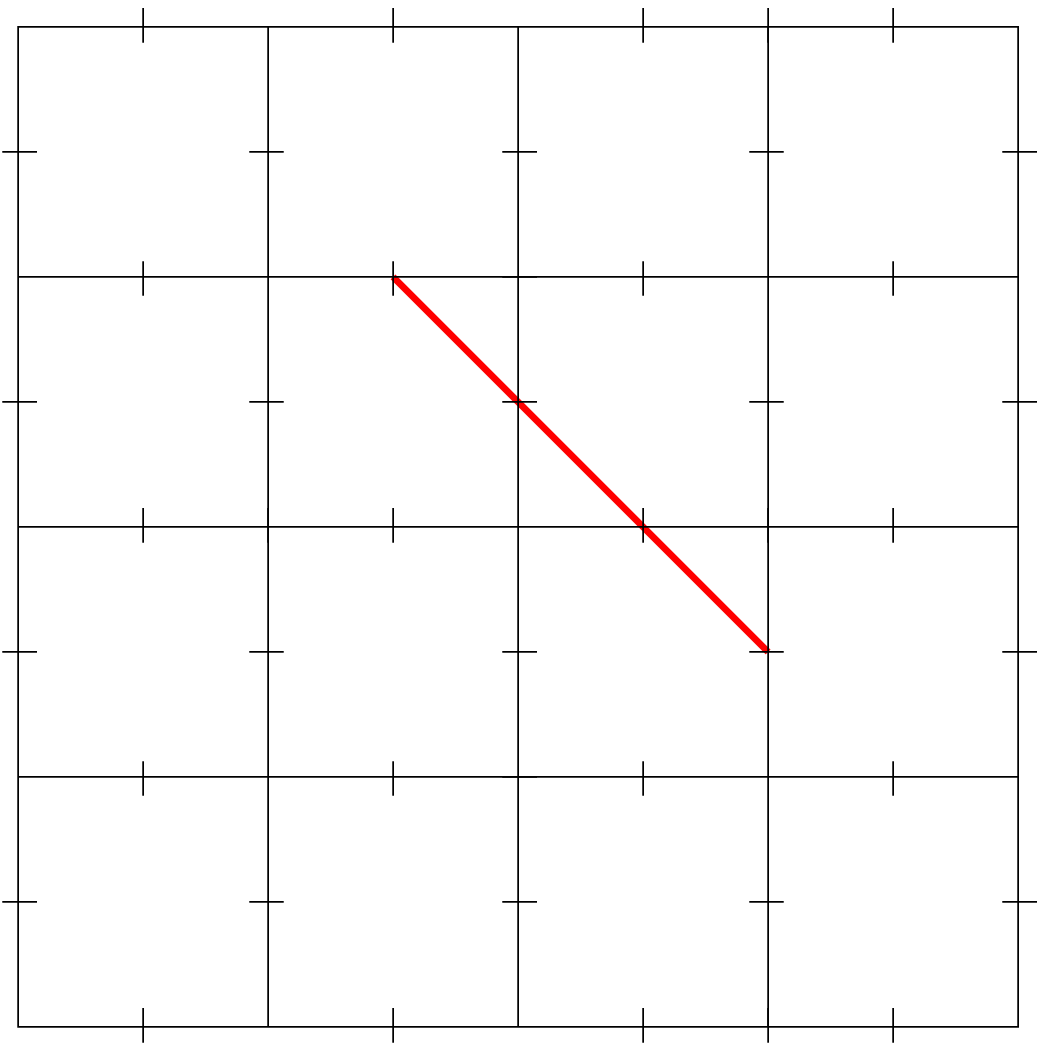}}%
  \caption{Example of 2D beams.}
  \label{fig:2D-beams}
\end{figure} 


We define neighborhood relationships between beams as we did for beamlets in (\ref{eq:beamlet-neigh}); see Figure \ref{fig:2D-beams-neigh}.
Specifically, two beams are neighbors if they are of the form
$[b_1,b_2]$ and $[b_2,b_3]$ and satisfy
\begin{eqnarray}
&& \|(b_{r,2} - b_{r,1})(b_{3} - b_{2}) - (b_{r,3} - b_{r,2})(b_{2} - b_{1})\| \nonumber \\
&& \qquad \qquad \qquad \qquad \qquad \qquad \leq (11\delta/20) (\|b_{3} - b_{2}\|  + \|b_{2} - b_{1}\|), \quad \forall r = 1, \dots, d. \label{eq:beam-neigh}
\end{eqnarray}
For example, suppose $[b_1,b_2]$ and $[b_2,b_3]$ are both $r_1$-beams.
If $b_3^*$ is the intersection of the line $(b_1,b_2)$ with the $r_1$-hyperplane $b_3$ belongs to, then condition (\ref{eq:beam-neigh}) implies $\|b_3 - b_3^*\| < (5/2) \delta$, so that $b_3$ is among the $5^{d-1}$ $r_1$-gridpoints closest to $b_3^*$.  

\begin{figure}[htbp]
  \centering
  \subfloat[Two $h1$-beams]{%
    \includegraphics[height = 2in]{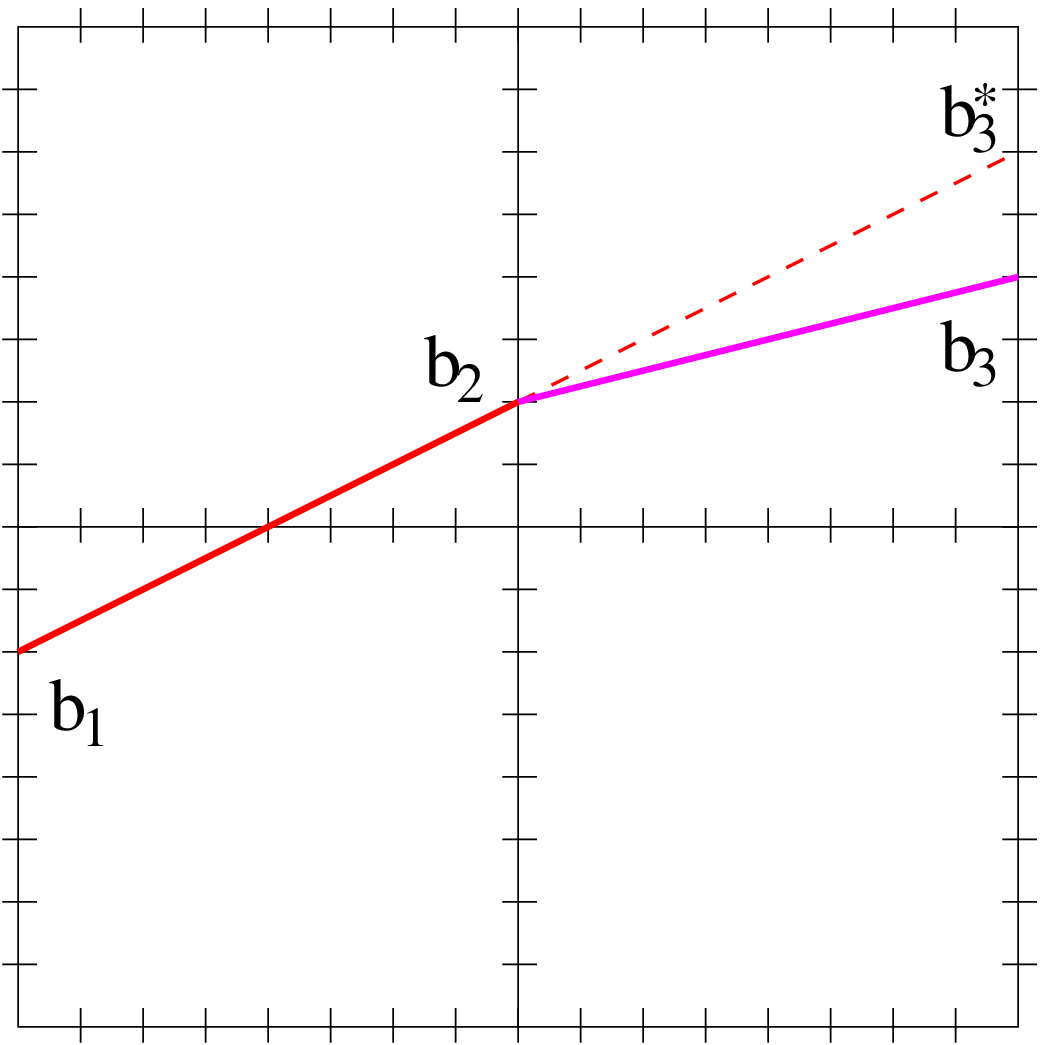}}%
  \quad%
  \subfloat[$h1h2$-beam and $r_1$-beam]{%
    \includegraphics[height = 2in]{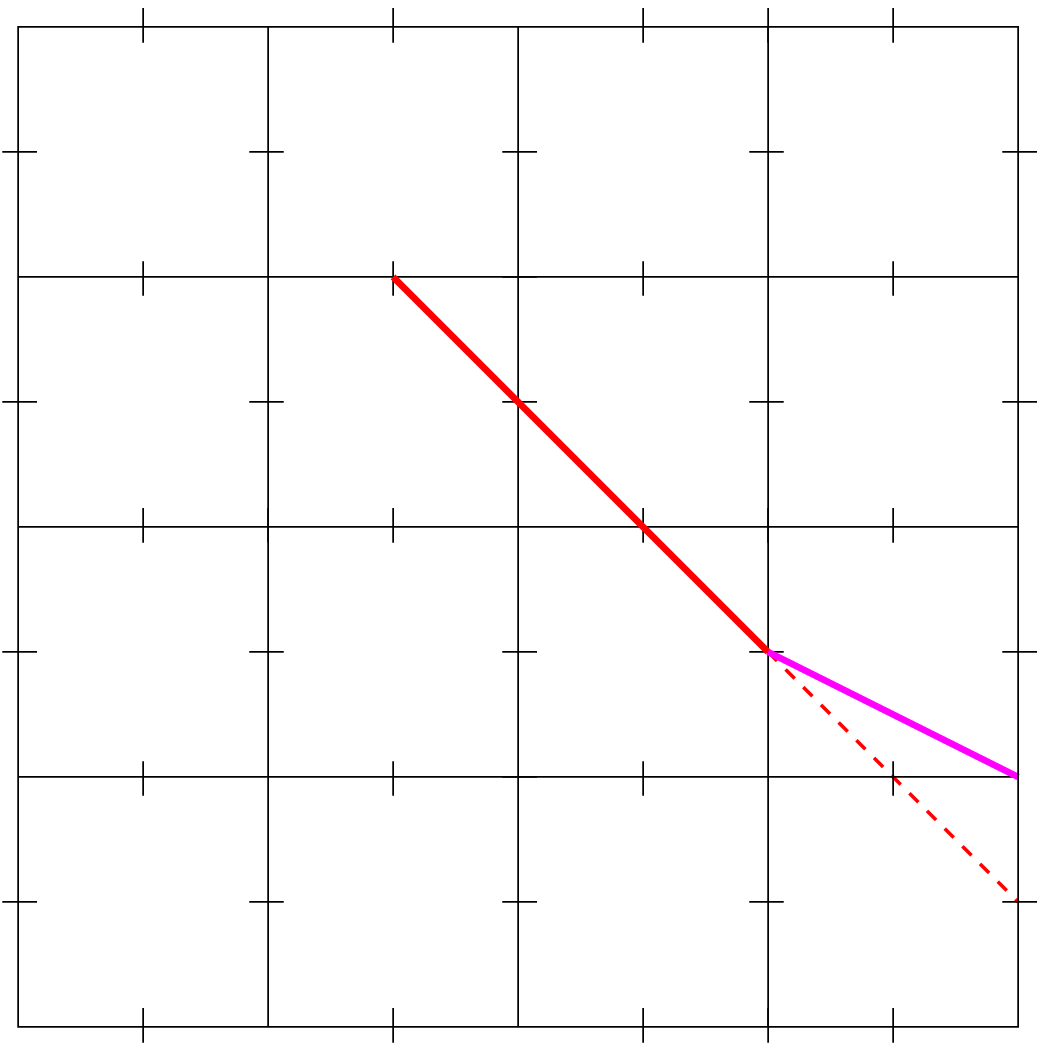}}%
  \quad%
  \subfloat[Two $h1h2$-beams]{%
    \includegraphics[height = 2in]{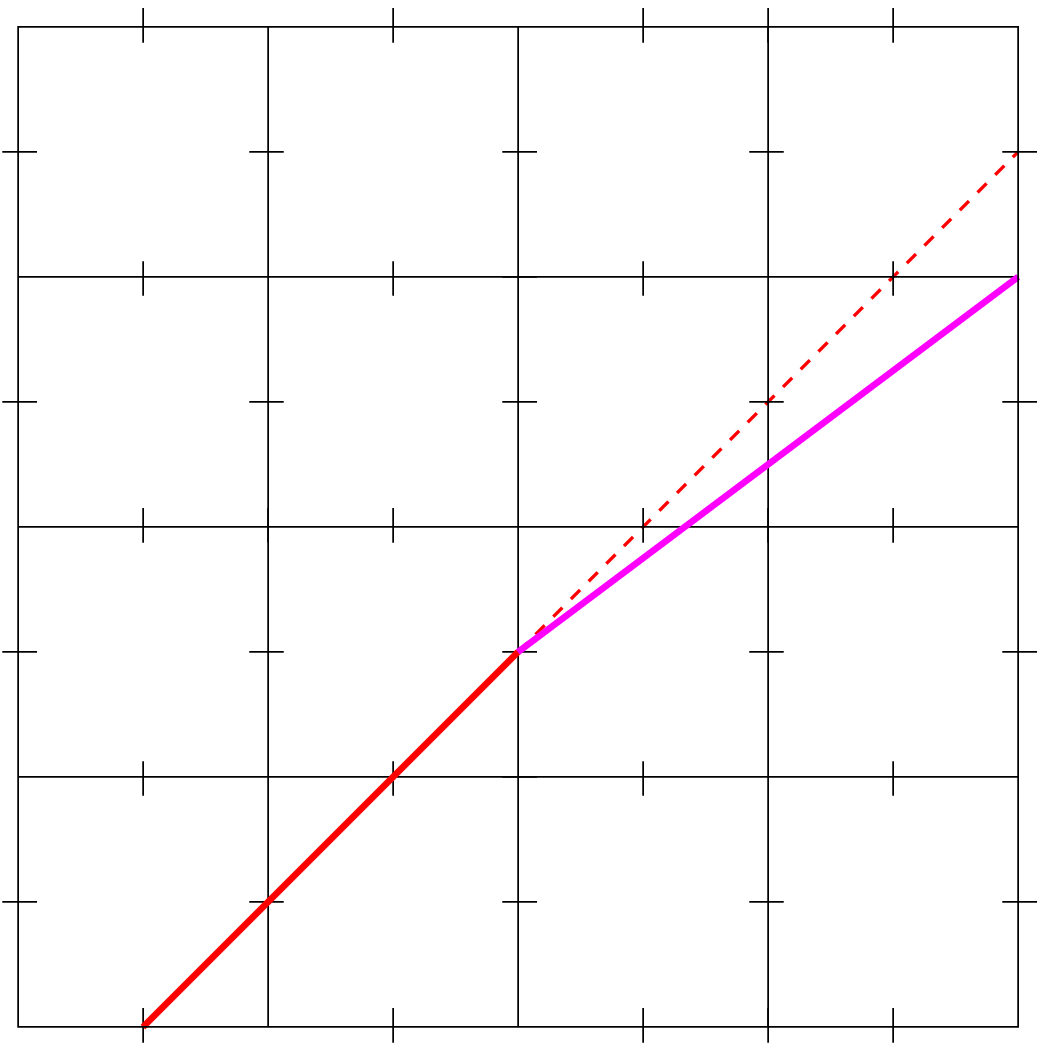}}%
  \caption{Example of 2D beams in good continuation.}
  \label{fig:2D-beams-neigh}
\end{figure} 


Let $\overline{\bbB}_{j,J}^d$ denote the resulting beam good-continuation network at scale $j$.   

\begin{lemma}
\label{lem:beam-count}
The number of nodes in $\overline{\bbB}_{j,J}^d$ is of order $O(d 2^{2(d-1)J - (3d-4)j} (1 + d 2^{2j - J}))$.
Moreover, all nodes have at most $2 d \cdot 7^{d-1}$ neighbors, with most nodes having at most $2 \cdot 5^{d-1}$ neighbors.
\end{lemma}
\begin{proof}
The proof of Lemma \ref{lem:beam-count} is in the Appendix.\end{proof}

Just as we did for beamlets, to each beam $B$ we associate a tubular region:
$$R(B) = \{\bx \in [0,1]^d: \min_{\by \in B} \|\bx - \by\| \leq \delta\}.$$

\begin{theorem}
\label{th:curve-covering}
Fix $\alpha \in (1,2]$ and $\lambda,\kappa>0$.
There is a universal constant $K$ such that, when $J$ is large enough and $j \geq j(\alpha,\beta) + K$, to each curve $\gamma \in \Gamma(\alpha,\lambda,\kappa)$ corresponds a path $\pi_{j}(\gamma)$ in $\overline{\bbB}_{j,J}^d$ chaining at most $\lambda 2^j + 2$ beams such that $\gamma$ is included in $R(\pi_j)$.
\end{theorem}
\begin{proof}
The proof of Theorem \ref{th:curve-covering} is in the Appendix.
\end{proof}

\begin{figure}[htbp]
  \centering
\includegraphics[height = 2.5in]{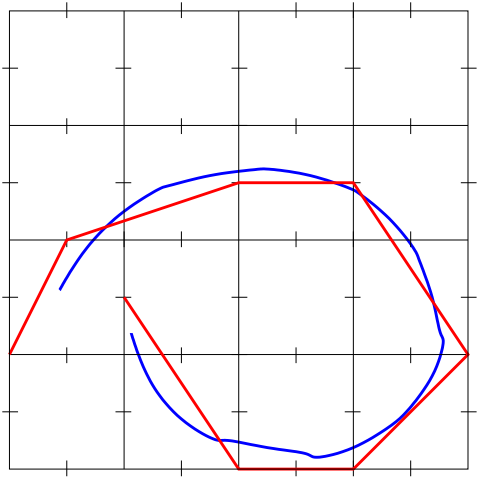}
  \caption{Example of a curve (blue) approximation by a chain of 2D beams (red) in good continuation.}
  \label{fig:2D-beam-approx}
\end{figure}

We now explain how Theorem \ref{th:curve-covering} implies Theorem \ref{th:beamlet-curve-covering}.
First, any beam may be approximated within distance $\delta_0$ by a chain of at most $2d$ beamlets at the same scale.
Indeed, a beam touches at most $2d$ $\Delta$-hyperplanes (at most $d+1$ for a typical $r$-beam); on each one, select a $(\Delta,\delta_0)$-gridpoint closest to the beam.
By successively connecting those gridpoints with line-segments, a chain of beamlets is born.
Therefore, a chain of beams of length at most $\lambda 2^j + 2$ may be approximated by a chain of beamlets of length at most $(2d)(\lambda 2^j + 2)$.
That the successive beamlets in such a chain are in good continuation according to (\ref{eq:beamlet-neigh}) comes from the fact that the beams themselves are in good continuation according to (\ref{eq:beam-neigh}), how beamlets are chained to approximate a beam, and the triangle inequality, in parallel to how Claim 3 is established in the proof of Theorem \ref{th:curve-covering}.

\begin{remark}
Both $\overline{\bbB}_{j,J}^d$ and $\bbG_{j,J}^{1,d}$ provide $\eps$-nets for H\"older curves embedded in the $d$-dimensional unit hypercube, with $\eps$ of order $2^{j-J}$. 
Which one is more economical?  
In terms of number of nodes, $\log_2(|\bbG_{j,J}^{1,3}|) \sim 2dJ - (3d-1)j$ (Lemma \ref{lem:hold-count}), while $\log_2(|\overline{\bbB}_{j,J}^d|) \sim 2(d-1)J - (3d-4)j$ (Lemma \ref{lem:beam-count}); the latter is smaller for all relevant scales $j \leq J/2$.  
Perhaps more importantly, $\bbG_{j,J}^{1,d}$ is substantially more connected, with most nodes with $2 \cdot 36^d$  neighbors (2592 in 2D; 93312 in 3D), compared to $\overline{\bbB}_{j,J}^d$, with most nodes having $2 \cdot 5^d$ neighbors (50 in 2D; 250 in 3D). 
\end{remark}

\section{Numerical Experiments}
\label{sec:experiments}

We perform some numerical experiments showing the power of approximation networks built on good-continuation principles.
Specifically, we compare the filamentary content of simulated 3D datasets with a variety of beamlet-based algorithms of our own creation.
Using software developed in \cite{shas} and a basic notion of good-continuity as introduced in Section \ref{sec:beamlets} is enough to outperform simpler algorithms disregarding any spatial information, such as introduced in \cite{spie-beamlets}.

We consider three situations illustrated in Figure \ref{fig:input_images}, where each column corresponds to a different case, and for each case the goal is to distinguish between the top and bottom situations.  
In setting (a), we compare a random point cloud (top) with a set of random filaments of different lengths, orientations and curvatures (bottom).
In setting (b), we compare a set of random filaments (top) with a set of random filaments constrained to pass through a small number of hubs (bottom). 
In setting (c), we compare a set of short random filaments (top) with a set of long random filaments (bottom), all filaments oriented in the direction of the first coordinate.

\begin{figure}[htbp]
\centering
\includegraphics[trim = 0mm 0mm 0mm 0mm, clip=true,width=0.8\textwidth]{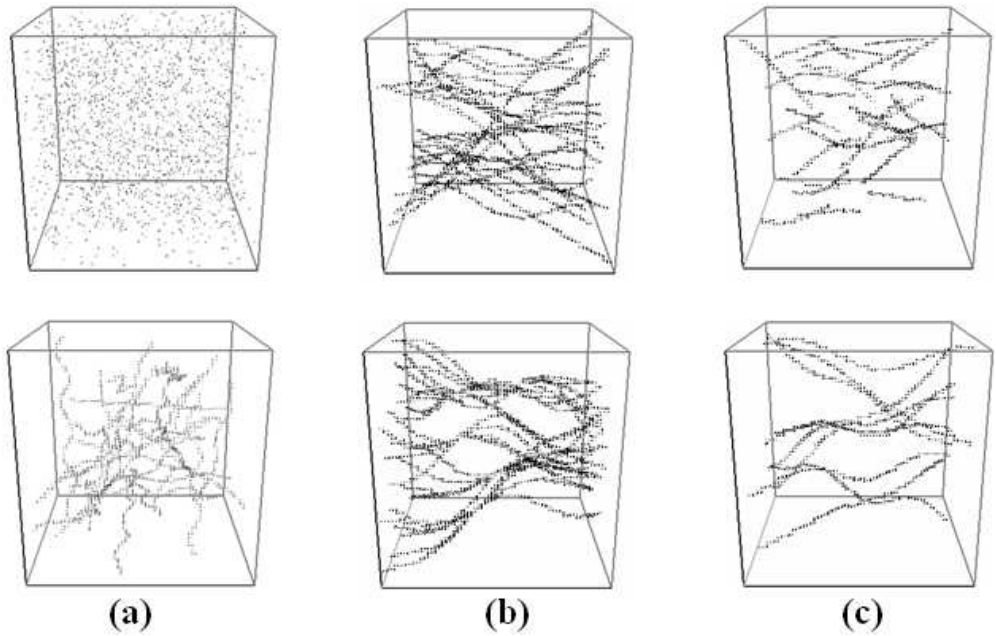}
\caption{Simulated datasets: (a) random point cloud (top) vs. random filaments (bottom); (b) random filaments (top) vs. random filaments with hubs (bottom); (c) random short filaments vs. random long filaments.}
\label{fig:input_images}
\end{figure}

Each dataset is a pixel array of size $64^3$.
The images are corrupted with a certain amount of additive white Gaussian noise (see Figure \ref{fig:noisy}), calibrated so the algorithm introduced in \cite{spie-beamlets} (and described in Section \ref{sec:LSI}) is powerless, i.e. essentially useless at distinguishing between the top and bottom settings. 
All filaments in these datasets are synthesized using trigonometric functions, each with randomly selected location, amplitude, frequency and phase shift.
   
\begin{figure}[htbp]
\centering
\begin{tabular}{cc}
\includegraphics[trim = 0mm 0mm 0mm 0mm, clip=true,width=0.45\textwidth]{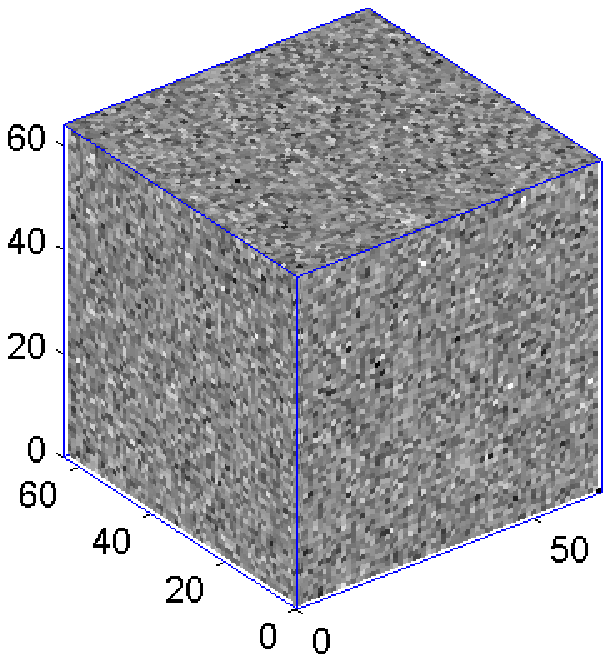} & \includegraphics[width=0.45\textwidth,trim = 0mm 0mm 0mm 0mm, clip=true]{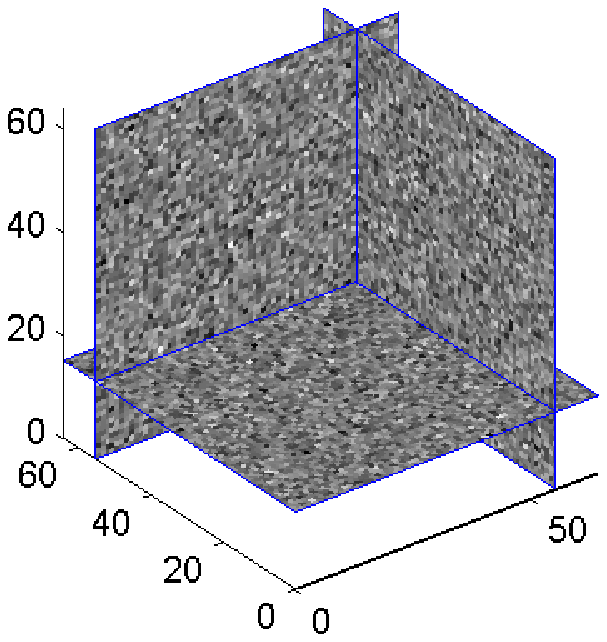}\\
(a) Surface of the noisy data volume & (b) Orthogonal slices through the noisy data
\end{tabular}
\caption{Noisy version of the simulated data in the lower panel of column (a) in Figure \ref{fig:input_images}.}
\label{fig:noisy}
\end{figure}

Though we introduce a variety of statistics, the workflow is the same:
\begin{enumerate} \setlength{\itemsep}{1pt}
\item {\it Compute the beamlet transform.}  
This amounts to computing all the beamlet coefficients, where a beamlet coefficient is defined as the line-integral of the image along the beamlet \cite{2D-beamlets, 3D-beamlets}.
\item {\it Thresholding of the beamlet coefficients.}  
Beamlets with a large coefficient are suspected to be intersecting with a filament, and therefore informative since our objective involves detecting filaments.
We therefore focus on those beamlets with large coefficients by discarding those with low coefficients \cite{2D-beamlets, 3D-beamlets}.
The choice of threshold is not determined in advance, rather a variety of thresholds are considered within some range.    
\item {\em Construction of a good-continuation network (GCN) for each scale.}  
Based only on the beamlets surviving thresholding, neighboring relationships as introduced in Section \ref{sec:beamlets} are considered.
\item {\em Extraction of some relevant network statistics.}
We extract a number of statistics from the beamlet network that are sensitive to different degrees and kinds of filamentarity such as node and edge cardinality or connected components counts, centrality measures or even more sophisticated statistics based on path search in the network.
\end{enumerate}
Our software is based on the beamlet transform as implemented in \cite{shas} and graph algorithms from the LEDA library \cite{leda}.
Our code is available online \cite{code}. 

In what follows, we will say that a test or method is `powerful' if it faithfully (more than 95\% of the time) distinguishes between the top and bottom settings.

\subsection{Number of Edges vs. Number of Nodes}
\label{section:edges_by_nodes}
Our simplest algorithm looks at how the number of edges and the number of nodes vary together as functions of the threshold used.

Consider column (a) of Figure \ref{fig:input_images}, where we compare a random point cloud (top) with a set of filaments with random lengths, orientations and curvatures (bottom).
In Figure \ref{fig:edges_by_nodes}, the number of edges is plotted as a function of the number of nodes for both top and bottom datasets.
As expected, the curve corresponding to the random filaments rests above the curve corresponding to the random point cloud, since the GCN is more connected in the former setting.
From our simulation studies, we found that this statistic is not powerful at per pixel signal-to-noise ratios (SNRs) below 0.8 for this specific situation.
The dataset with filaments contains 20 of them with random lengths in the range $[10,64]$ (in number of pixels).

\begin{figure}[htbp]
\centering
\includegraphics[trim = 0mm 0mm 0mm 5mm, clip=true,width=0.46\textwidth,height=0.35\textwidth]{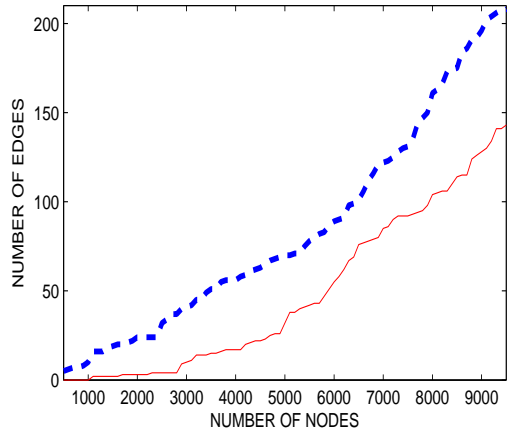}
\caption{Number of edges vs. number of vertices for the random point cloud (red, solid curve) and the dataset of random filaments (blue, dashed curve) from column (a) of Figure \ref{fig:input_images}. }
\label{fig:edges_by_nodes}
\end{figure}

\subsection {Vertex Betweeness}
Centrality attributes are well-known from network analysis \cite{networks_structure} and used for example in social networks studies \cite{network_social}.  The \emph{betweeness} of vertex $v$ is defined as
$$
C_B(v) =\sum_{s\neq v \neq t}\frac{\rho_{st}(v)}{\rho_{st}},
$$
where $s,t$ run through all vertices except $v$, $\rho_{st}$ is the number of longest paths from $s$ to $t$, and $\rho_{st}(v)$ the number of longest paths from $s$ to $t$ that pass through a vertex $v$.
(Note that shortest paths are sometimes used instead.) 

Consider column (b) in Figure \ref{fig:input_images}, where the top image contains randomly distributed filaments and the bottom image also contains random filaments but with hubs.  
The number of filaments and their individual characteristics are the same in both cases. 
In Figure \ref{fig:betweeness_cum} we plot the number of vertices with betweeness exceeding a given threshold for both datasets.
As expected, the curve corresponding to the random filaments with hubs rests above the curve corresponding to the random filaments without hubs, since in the former case the hubs translate into vertices in the GCN with large betweeness.
From our simulation studies, we also found this statistic not powerful at SNRs below 0.8 for the type of filaments chosen here.
The data set without hubs (top) contains 20 filaments with horizontal orientations and random lengths in range $[60,63]$. The data set with hubs contains 5 groups of 4 filaments with horizontal orientations and length 64, then all 4 filaments in the group have common filamentarity region of length 3 (hub).  

\begin{figure}[]
\centering
\includegraphics[trim = 0mm 0mm 0mm 0mm, clip=true,width=0.5\textwidth]{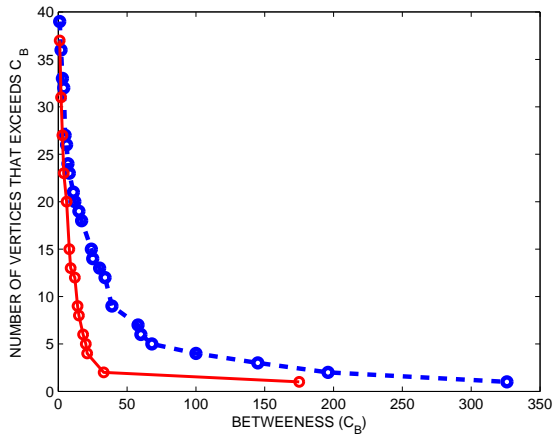}
\caption{Number of vertices whose betweeness exceeds a given threshold for the dataset of random filaments without hubs (red, solid) and the dataset of random filaments with hubs (blue, dashed) from column (b) of Figure \ref{fig:input_images}.}
\label{fig:betweeness_cum}
\end{figure}
 
Now, consider a slightly different setting with just one large hub, illustrated in the left column of Figure \ref{fig:max_betweeness_snr}.  
In the right column of the Figure \ref{fig:max_betweeness_snr} we compare the maximal betweeness of both datasets for different noise levels, which is significantly larger for the dataset with a hub.
Again, we found that the maximal betweeness statistic is not effective when the voxel-level SNR is less than $0.8$.
The data set without hub (top) contains 18 filaments with horizontal orientations and random lengths in range $[48,64]$. The data set with hub contains 18 filaments with horizontal orientations and length 64, then all filaments in the group have common filamentarity region of length 5 (hub).  

\begin{figure}[htbp]
\centering
\begin{tabular}{cc}
\includegraphics[trim = 0mm 3mm 0mm 7mm, clip=true,height=0.4\textwidth]{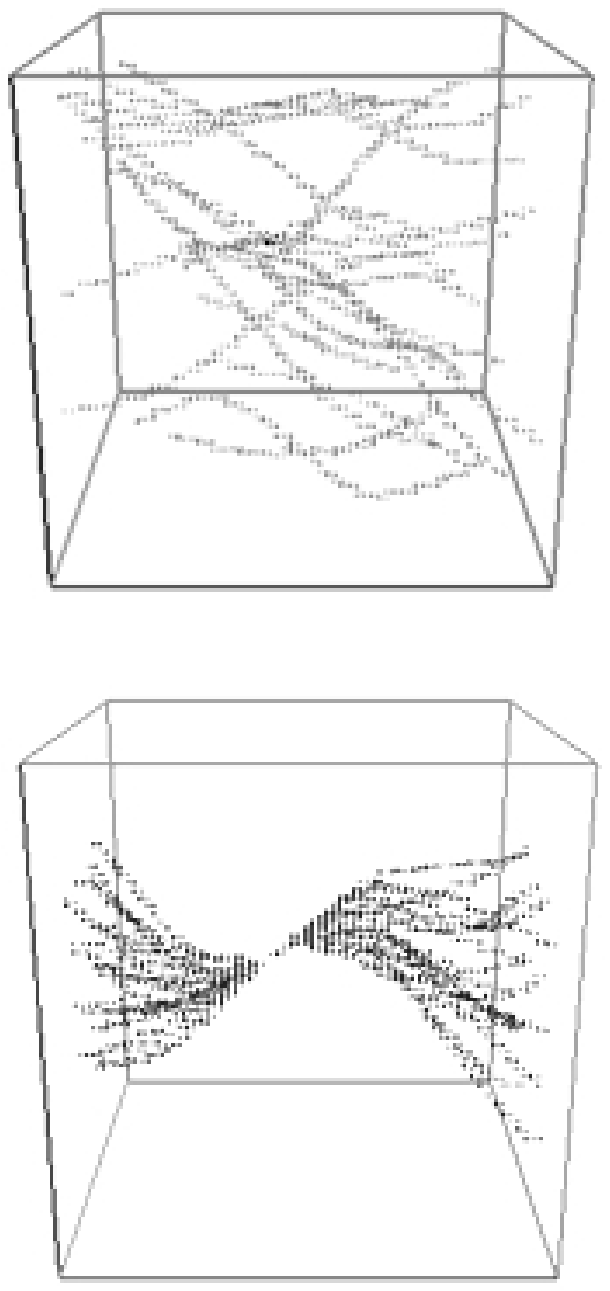} & \includegraphics[height=0.4\textwidth,trim = 0mm 3mm 0mm 0mm, clip=true]{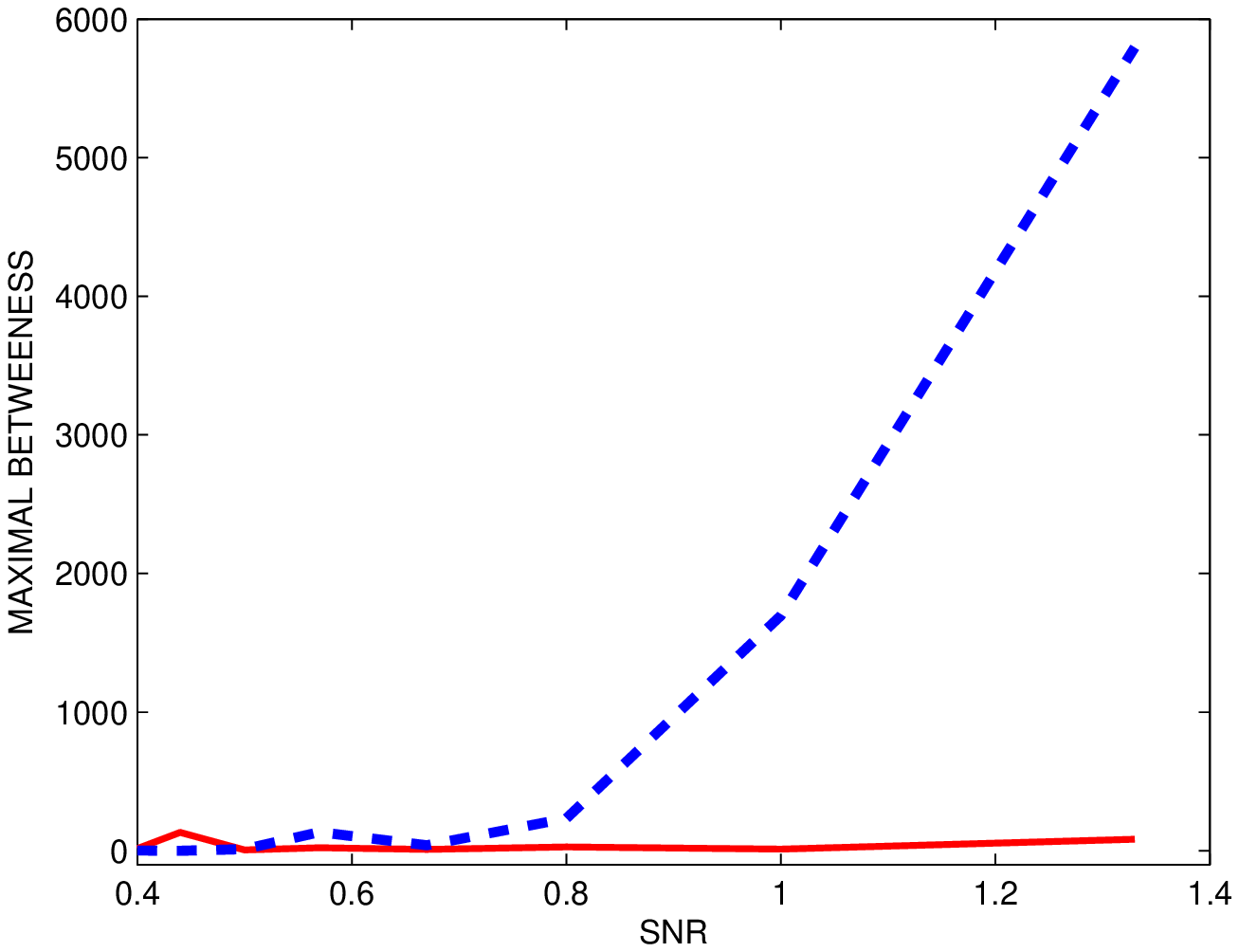}
\end{tabular}
\caption{Maximal betweeness as a function of SNR, for the dataset of random filaments without hub (red, solid) and the dataset of random filaments with a hub (blue, dashed) from the left column.}
\label{fig:max_betweeness_snr}
\end{figure}

\subsection {Filamentarity Survival Index}

For each beamlet $v$ in the network, let $\omega(v)$ denote its coefficient (or weight), i.e. the line-integral of the image along this beamlet. 
The weight of a beamlet path $p$ is then defined as the sum of the weights of the beamlets it passes through:
$$\omega(p)=\sum_{v\in p}\omega(v).$$ 
Let $P$ be the partitioning of $\bbG$ into disjoint paths defined recursively as follows:
\begin{equation}\label{eq_lp}
\left\{ \begin{array}{l l}
   p_1 = \arg\max_{p\in\bbG}\omega(p);\\ 
   p_i = \arg\max_{p\in\bbG\backslash  \{p_1,\ldots,p_{i-1} \}}\omega(p), & \quad i>1.
\end{array}\right.
\end{equation} 
We define the \emph{Filamentarity Survival Index} $D(t)$ as the fraction of paths in $P$ with weight greater than $t$.
The computation of the piecewise constant function  $D(t)$ is iterative, where at each iteration the \emph{Longest Weighted Path} (LWP) is found and removed from the network. The algorithm terminates when the network is empty.

We next define the \emph{Filamentarity Survival Ratio} (FSR) which compares the filamentary survival index of a given dataset $I$ with  the filamentary survival index a random point cloud with same energy, denoted $\tilde{I}$
$$
R(t)= \frac{\left(D_{I}\left(t\right)+\epsilon \right)}{\left(D_{\tilde{I}}\left(t\right)+\epsilon \right)},
$$
where $\epsilon$ is a small positive number, used to avoid divide-by-zero situations. 

Consider column (c) in Figure \ref{fig:input_images}, where we compare a dataset of random short filaments (top) with a dataset of random long filaments (bottom), all filaments oriented in the direction of the first coordinate.
Figure \ref{fig:FSR_1} presents the FSR curves for these two datasets.  Additionally, this figure contains the FSR curve for the case of purely random points, as in the top image in column (a) in Figure \ref {fig:input_images}.  
As expected, the FSR curve for the purely random dataset (black, dotted line) remains close to 1.  
For the dataset containing long filaments, however, the FSR (blue, dashed line) is significantly higher than 1 for certain values of $t$. 
For the dataset containing short filaments, the FSR curve (red, solid line) is found between the other two curves.
We found that the FSR statistic is powerless at distinguishing between short and long filaments (in this setting) for SNRs below 1 in this specific situation.    
The data set with short filaments (top) contains 30 filaments with horizontal orientations and lengths 20. The data set with long filaments (bottom) contains 10 filaments with horizontal orientations and lengths 60. 

\begin{figure}[htbp]
\centering
\includegraphics[width=0.6\textwidth]{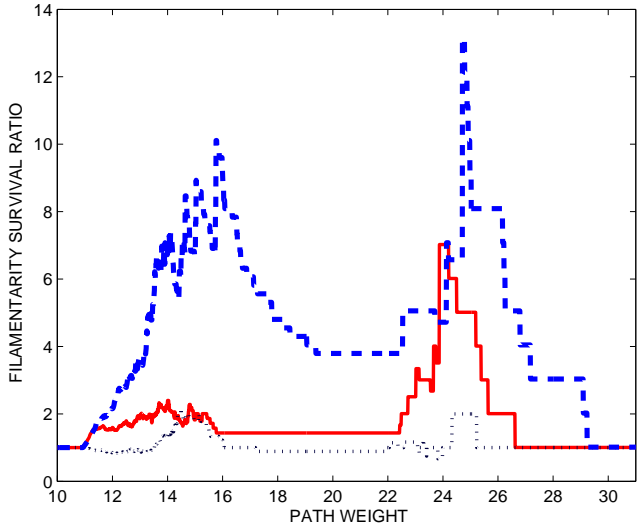}
\caption{Filamentarity Survival Ratios for the dataset of random short filaments (red, solid) and for the dataset of random long filaments (blue, dashed) in column (c), and for the random point cloud (black, dotted) in column (a) of Figure \ref{fig:input_images}.}
\label{fig:FSR_1}
\end{figure}

\subsection{Discussion}

\subsubsection{Comparison with statistics based on beamlet coefficients only}
\label{sec:LSI}

The numerical experiments we presented provide evidence that using the spatial relationship between beamlets allows for the design of algorithms that have the ability to perform well even in the case of very low SNR, where statistics based only on the beamlets coefficients fail.  
More precisely, we compared our different algorithms described above with the \emph{Log-Survival Index} (LSI) introduced in \cite{spie-beamlets}; see Table \ref{tab:summary_exp}. 
The LSI is defined as
$$
S_j(t)=\frac{\log(1+ N_j(t) )}{\log(1+N_j)},
$$
where $N_j(t)$ is the number of beamlet coefficients at the $j$-th scale that exceed $t$, and $N_j$ is the total number of beamlet coefficients at scale $j$.\\ 

\begin{table}[h]
\centering
\begin{tabular}{|c|c|c|c|}
\hline
Case  & (a) & (b) & (c) \\ \hline
Statistic & Graph connectivity &
Vertex betweeness & FSR \\ \hline
Lowest SNR & \quad 0.8\ (1.33) \quad & \quad 0.8\ (2) \quad & \quad 1\ (1.33) \quad \\ \hline
\end{tabular}
\caption{Summary of the experiments.  The numbers (resp. numbers in parentheses) in the last row correspond to the lowest SNR at which the statistic used for that particular case (resp. the LSI statistic) is still powerful.
}
\label{tab:summary_exp}
\end{table}

Figure \ref{fig:lsis} shows the behavior of this statistic for situation (a) of Figure \ref{fig:input_images}. 
The LSI curves are clearly disjoint at SNR = 4; however, they merge at SNR = 0.8, resulting in the LSI being powerless, while the statistic presented in Section 5.1 is still powerful.

\begin{figure}[htbp]
  \centering
  \subfloat[SNR = 4]{\label{fig:lsi_working}%
    \includegraphics[trim = 0mm 11mm 0mm 0mm, clip=true,width=0.4\textwidth]{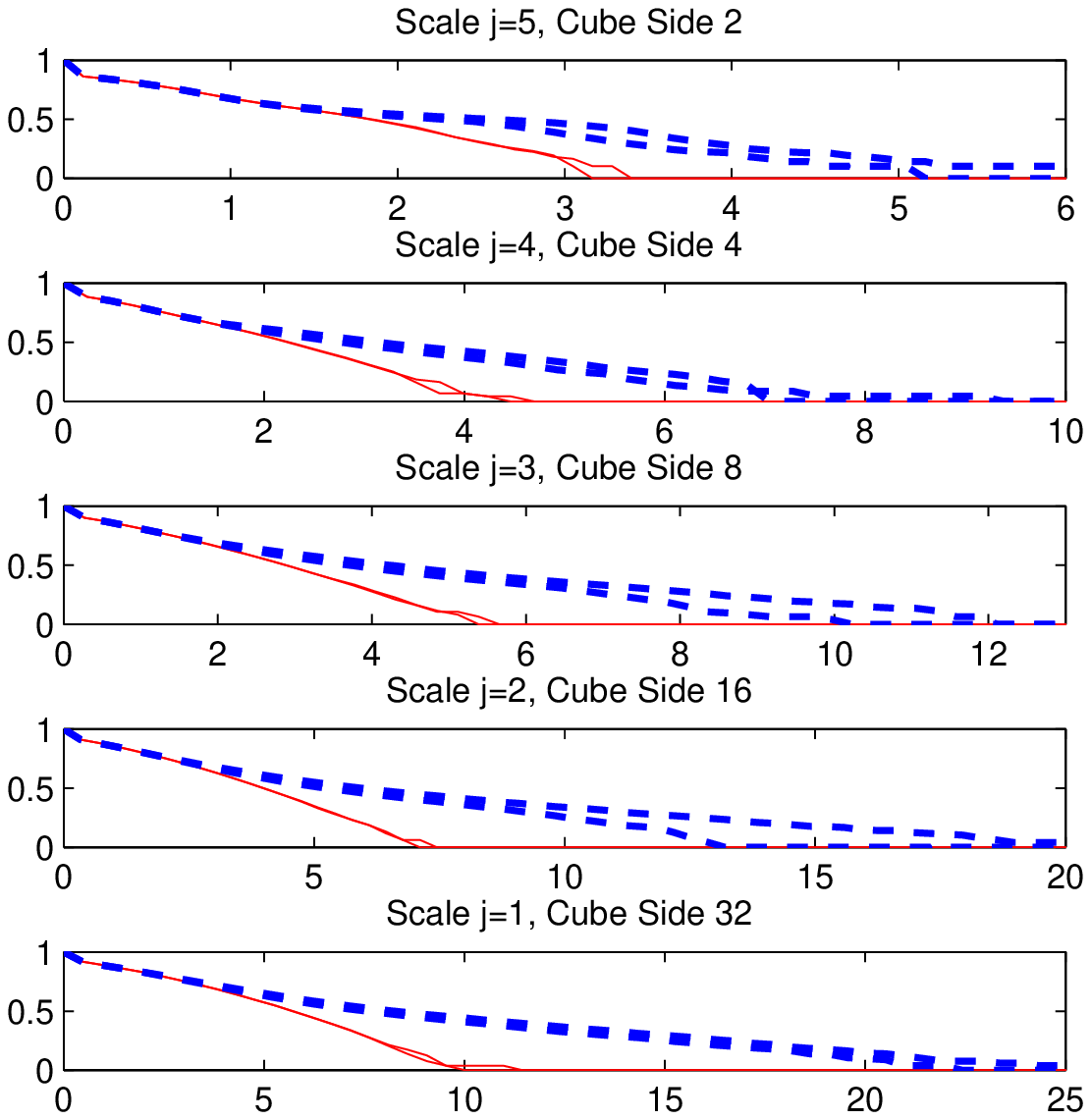}}%
  \quad%
  \subfloat[SNR = 0.8]{\label{fig:lsi_not_working}%
    \includegraphics[trim = 0mm 10mm 0mm 0mm, clip=true,width=0.4\textwidth]{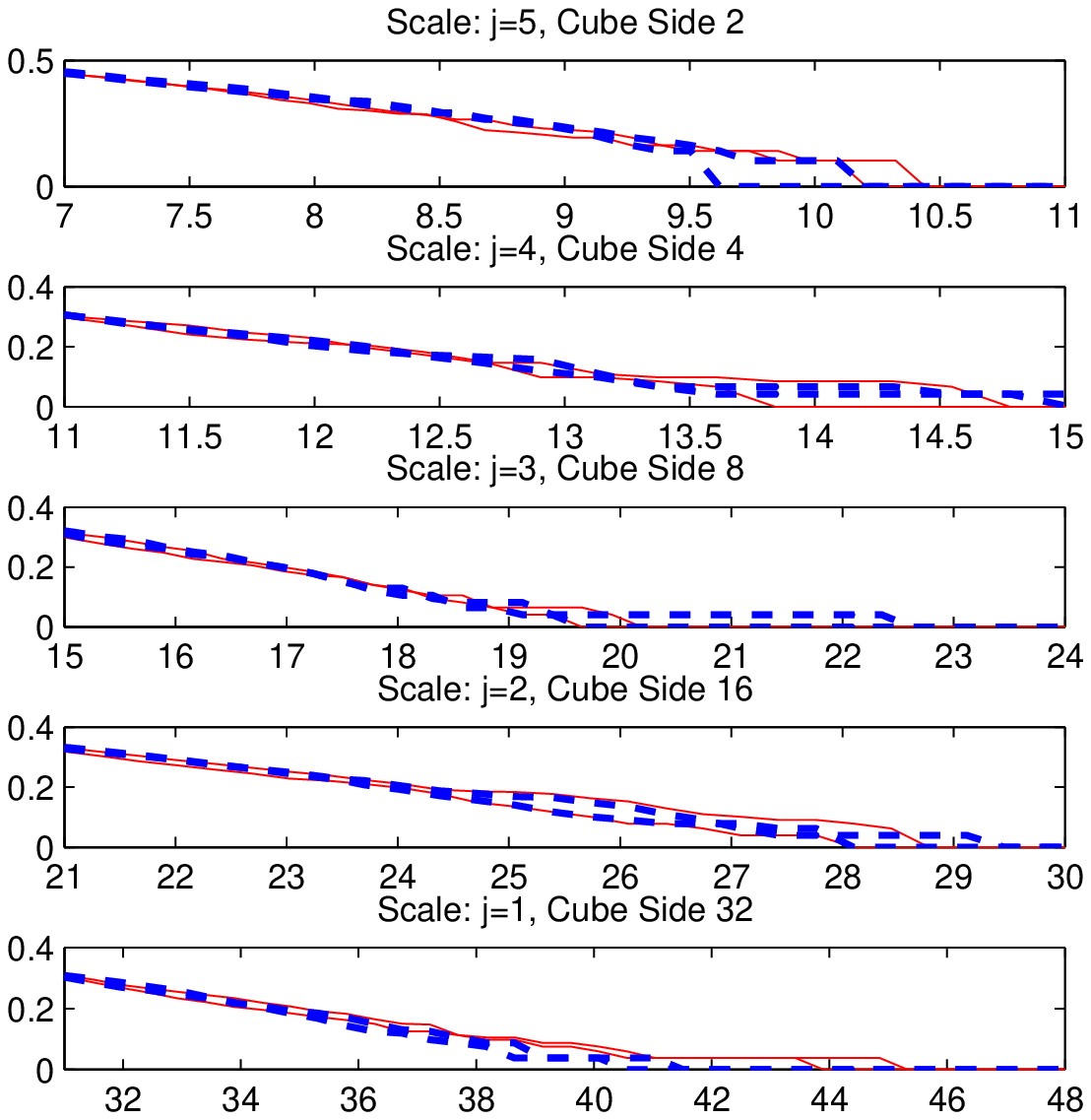}}%
  \caption{LSI curves at various scales for the random point cloud (red,  solid) and for the dataset with random filaments (blue, dashed) in column (a) of Figure \ref{fig:input_images}.}
  \label{fig:lsis}
\end{figure}

\subsubsection{Computational Complexity}

The computational burden comes from the beamlet transform. In these experiments, we used the version developed in \cite{shas}, which runs in order $O(n^5)$ flops for an $n^3$ pixel array.
The version based on the Fast Slant Stack \cite{3D-beamlets, FSS} is in theory faster, $O(n^4 \log n)$ flops, yet in practice the implementation in \cite{shas} is more precise and faster on smaller arrays as considered here.

We mention that beamlets have only been developed for 2D and 3D datasets, partly because the method suffers from the curse of dimensionality, since the number of beamlets increases exponentially with the ambient dimension; see Lemma \ref{lem:beamlet-count}.
Note also that the implementation in \cite{2D-beamlets} runs in $O(n^2 \log n)$ for an $n^2$ pixel array, so that the 3D implementations are comparatively heavier.  

\subsubsection{Computation of the Longest Weighted Paths}
Since finding the LWP is an NP-Hard problem \cite{longest-path-NP}, in the experiments described above we restricted ourselves to acyclic networks.  We assume that the foreground signal consists of a set of curves of the form $\gamma(s)=(s,\gamma_{y}(s),\gamma_{z}(s))$.
Therefore, it is possible to approximate such curves with beamlets oriented along the fist axis only, so their chaining will never induce a cycle.  This rather artificial assumption was made only for computational reasons; for general datasets one would be forced to use approximations \cite{longest-path-approx} or other strategies \cite{dp_tracking,2D-beamlets}.
We tried statistics based on connected components, which are computationally tractable, but were not able to improve on the LSI.


\subsubsection{Full multiscale analysis} 
We built a good-continuation network for each scale separately along the lines presented in Section \ref{sec:beamlets}.  
Based on Theorem \ref{th:curve-covering}, this is fine if we expect filaments of homogenous smoothness.  
If we want to test for filaments of varying smoothness, building a single good-continuation network that includes neighboring relationships across scales may be more useful.  

\vspace{.3in}
\begin{center}
{\Large \bf Appendix}
\end{center}




\subsection*{Proof of Lemma \ref{lem:hold-count}}
\label{proof:hold-count}
We first count the number of nodes $(\bm,\bh)$.
There are $\Delta^{-k}$ choices for $\bm$ and $2 \beta \delta_{|\bs|}^{-1}$ for each $h^{(\bs)}$, so for $\bh$ a total of 
$$\prod_{|\bs| \leq \afloor} 2 \beta \delta_{|\bs|}^{-1} = (2 \beta)^{c_3} \delta^{-c_3} \Delta^{c_4}.$$

For a given node $(\bm,\bh)$, $\bm$ has at most $2k$ neighbors $\bm_\star$'s in the square grid.
And for each one of them, there are at most $6^{c_3}$ $\bh_\star$'s satisfying the left part of (\ref{eq:hold_neigh}), since at each $\bs$ there are at most 6 choices. \QED

\subsection*{Proof of Theorem \ref{th:graph_cover}}
\label{proof:imm_cover}
In view of Lemma \ref{lem:local_approx}, it is enough to show that, for $\bm,\bm_\star \in \{1,\dots,\Delta^{-1}\}^k$ with $\bm_\star = \bm + \be_i$, 
$$\left|h^{(\bs)}(\bm_\star,f) - \sum_{t \leq \afloor - |\bs|}
  \frac{h^{(\bs+t \be_i)}(\bm,f)}{t!}\right| < 3.$$ 

By the triangle inequality, we have
\begin{eqnarray*}
\left|h^{(\bs)}(\bm_\star,f) - \sum_{t \leq \afloor - |\bs|}
  \frac{h^{(\bs+t \be_i)}(\bm,f)}{t!}\right| & \leq &
\left|h^{(\bs)}(\bm_\star,f) - \delta_{|\bs|}^{-1}
f^{(\bs)}(x_{\bm_\star})\right| + \\
& & \left|\delta_{|\bs|}^{-1} f^{(\bs)}(x_{\bm_\star}) - \sum_{t \leq \afloor - |\bs|}
\frac{\delta_{|\bs| + t}^{-1} f^{(\bs + t \be_i)}(x_{\bm})}{t!}\right| + \\
& & \sum_{t \leq \afloor - |\bs|} \frac{\left|\delta_{|\bs|+t}^{-1} f^{(\bs +
t \be_i)}(x_{\bm}) - h^{(\bs+t \be_i)}(\bm,f)\right|}{t!}.
\end{eqnarray*}

By definition of $h^{(\bs)}(\bm,f)$, the first term on the right handside
is bounded by 1/2 while the third term is bounded by $1/2\ \sum_{t} 1/t! \leq
\exp(1)/2 < 3/2$.
So we are left with showing that the second term is bounded by 1.
To do that, we use Lemma \ref{lem:taylor_ei} and the fact that $x_{\bm_\star,i} - x_{\bm,i} = \Delta$ to get
$$\left|f^{(\bs)}(x_{\bm_\star}) - \sum_{t} f^{(\bs +
t \be_i)}(x_{\bm})\ \frac{\Delta^t}{t!}\right| \leq c_1 \beta
\Delta^{\alpha-|\bs|}.$$
Since $\Delta^t = \delta_s \delta_{s+t}^{-1}$ for all $s,t$, we
further get
$$\left|\delta_{|\bs|}^{-1} f^{(\bs)}(x_{\bm_\star}) - \sum_{t}
\frac{\delta_{|\bs| + t}^{-1} f^{(\bs + t \be_i)}(x_{\bm})}{t!}\right| \leq c_1 \beta \Delta^\alpha/\delta.$$
This concludes the proof of Theorem \ref{th:graph_cover}. \QED

\subsection*{Proof of Theorem \ref{th:GLRT-approx}}
We use simplified notation for clarity.
By Boole's Inequality, we have
$$\pr{M_\cP > C \delta^{d-k} n | H_0} \leq |\cP| \cdot \max_{P \in \cP} \pr{N(\oR(P)) > C \delta^{d-k} n | H_0}.$$

For the number of paths,
\begin{equation} \label{eq:P-card}
|\cP| = O(\delta^{-(d-k)c_3} \Delta^{(d-k) c_4} \cdot (6^{(d-k) c_3})^{\Delta^{-k}}),
\end{equation} 
since there are $O(\delta^{-(d-k)c_3} \Delta^{(d-k) c_4})$ choices for the starting point and $6^{(d-k) c_3}$ at each step after that, along a path of length $\Delta^{-k}$ (see Lemma \ref{lem:hold-count-im}).
Hence, $\log(|\cP|) \leq a_1 \delta^{-k/\alpha}$ for some constant $a_1$ not depending on $\delta$.

Under $H_0$, $N(\oR(P))$ is binomial with parameters $n$ and  $|\oR(P)|_d$, so that
$$\max_{P \in \cP} \pr{N(\oR(P)) > C \delta^{d-k} n | H_0} = \pr{\Bin(n, \max_{P \in \cP} |\oR(P)|_d) > C \delta^{d-k} n}.$$
By integrating with respect to $\bz$ (the last $d-k$ coordinates) first, we have:
\begin{equation}
\label{eq:region-volume}
|\oR(\bm, \bh_1, \dots, \bh_{d-k})|_d = (c_0+1) \Delta^k \delta^{d-k}, \quad \forall (\bm, \bh_1, \dots, \bh_{d-k}).
\end{equation}
Hence, for all $P \in \cP$, $|\oR(P)|_d = (c_0+1) \delta^{d-k}$, and therefore, since $\delta = n^{-\alpha/(k+\alpha(d-k))}$,
$$\max_{P \in \cP} \pr{N(\oR(P)) > C \delta^{d-k} n | H_0} = \pr{\Bin(n, (c_0+1) \delta^{d-k}) > C \delta^{d-k} n}.$$
By standard large deviation bounds like Bernstein's Inequality \cite{ShoWel}, the logarithm of the right hand side is bounded from above by $- a_2 C \delta^{d-k} n$ for $C \geq 2 (c_0+1)$, where $a_2$ does not depend on $\delta$ or $n$.

Collecting terms, we get the following bound:
$$\log \pr{M_\cP > C \delta^{d-k} n | H_0} \leq a_1 \delta^{-k/\alpha} - C a_2 \delta^{d-k} n, \quad \forall C \geq 2 (c_0+1).$$
By choosing $C = 2 ((c_0+1) \vee (a_1/a_2))$, the right hand side is negative if $\delta \geq n^\rho$.
\QED

\subsection*{Proof of Theorem \ref{th:LSRT}}
We use simplified notation for clarity.
By (\ref{eq:region-volume}), for all $(\bm, \bh_1, \dots, \bh_{d-k}) \in \bbG^{k,d-k}(\alpha, \beta)$, we have:
$$\pr{S(\bm, \bh_1, \dots, \bh_{d-k}) = 1|H_0} = \pr{\Bin(n, a_1 \delta^{k/\alpha + d-k}) > \tau n \delta^{k/\alpha + d-k}},$$
where $a_1 = (c_0+1) (c_1 \beta)^{-k/\alpha}$.
Let $q_0(\tau) = \max \pr{\Bin(n, a_1 p) > \tau n p}$ over $n \in \bbN$ and $p \in (0,1)$ such that $np \geq 1$; note that $q_0(\tau) \to 0$ as $\tau \to \infty$.
In the same way that we obtained (\ref{eq:P-card}), we find that the number of paths of the form 
$$\{(\bm_{\rm zz}(t), \bh_1(t), \dots, \bh_{d-k}(t)): t = t_0, \dots, t_0 + \ell -1\}$$ 
is bounded above by $a_2 \delta^{-(d-k)c_3} \Delta^{(d-k) c_4} \Delta^{-k} 6^{(d-k) c_3 \ell}$, for $a_2$ not depending on $\mu$ or $\eta$.
With this fact and Boole's Inequality, we get:
$$\pr{L > \ell | H_0} \leq a_2 \delta^{-(d-k)c_3} \Delta^{(d-k) c_4-k} 6^{(d-k) c_3 \ell} \cdot q_0^\ell.$$
Hence, with $q_0$ small enough (i.e. $\tau$ large enough),
we have
$$\pr{L > \log(1/\delta)|H_0} \to 0, \quad n \to \infty,$$
where we used (\ref{eq:Delta-delta}).
We then conclude with the fact that $\delta = n^{-\alpha/(k+\alpha(d-k))} \vee \eta$.

Assume for concreteness that $\eta > 0$.
Under $H_1$, let $P^* \in \cP$ such that $\graph_\eta(f^*) \subset \oR(P^*)$.
Note that $|\graph_\eta(f^*)|_d = \eta^{d-k}$ and $|\graph_\eta(f^*) \cap \oR(\bm,\bh_1, \dots, \bh_{d-k})|_d = \Delta^k \eta^{d-k}$ for all $(\bm,\bh_1, \dots, \bh_{d-k}) \in P^*$.
Therefore, together with the behavior under the null, for all $(\bm, \bh_1, \dots, \bh_{d-k}) \in P^*$, we have:
$$\pr{S(\bm, \bh_1, \dots, \bh_{d-k}) = 1|H_1} \geq \pr{\Bin(n, \eps_n \Delta^{k}) > \tau n \delta^{k/\alpha + d-k}}.$$
Hence, with $\eps_n > C n \delta^{d-k}$ we have, for all $(\bm,\bh_1, \dots, \bh_{d-k}) \in P^*$,
$$\pr{S(\bm,\bh_1, \dots, \bh_{d-k}) = 1|H_1} \geq \pr{\Bin(n, C(c_1 \beta)^{-k/\alpha} n \delta^{k/\alpha + d-k}) > \tau n \delta^{k/\alpha + d-k}}.$$
With $\tau$ fixed as above, let $q_1(b) = \min \pr{\Bin(n, b p) > \tau np}$ over $n \in \bbN$ and $p \in (0,1)$ such that $np \geq 1$; note that $q_1(b) \to 1$ as $b \to \infty$.
Assume that $n$ is large enough that $n^\rho > 1$, so that $n \delta^{k/\alpha + d-k} > 1$.
Then, for $C = b (c_1 \beta)^{k/\alpha}$,
$$\pr{S(\bm,\bh_1, \dots, \bh_{d-k}) = 1|H_1} \geq q_1(b), \quad \forall (\bm,\bh_1, \dots, \bh_{d-k}) \in P^*.$$
Now, by the Erd\"os-R\'enyi Law \cite{MR1092983} (in fact a modified version allowing for weak dependencies, found in Appendix A.3 in \cite{ery-thesis}), the longest significant run along $P^*$ has length at least $\log(|P^*|)/\log(1/q_1)(1+o(1))$ with high probability.
Hence, for $q_1$ large enough (i.e. $C$ large enough),
$$\pr{L > \log(1/\delta)|H_1} \to 1, \quad n \to \infty,$$
where we used $|P^*| = \Delta^{-k}$ together with (\ref{eq:Delta-delta}).
\QED

\subsection*{Proof of Lemma \ref{lem:beamlet-count}}
There are $O(\Delta^{-d})$ $\Delta$-hypercubes and each one of them has $O(\Delta \delta_0^{-1})^{d-1}$ gridpoints on anyone of its $2d$ faces.
Therefore, there are $O(\Delta^{-d} d^2 (\Delta \delta_0^{-1})^{2(d-1)}) = O(d^2 2^{2(d-1)J - (d-2)j})$ beamlets at scale $j$.

We now look at the degree of a beamlet $B = [b_1, b_2] \in \bbB_{j,J}^d$.
Let $r$ be such that $|b_{r,2} - b_{r,1}| = \|b_2 - b_1\|$.
Consider another beamlet of the form $[b_2,b_3]$, and define 
$$b_3^* = b_{2} + \frac{b_{r,3} - b_{r,2}}{b_{r,2} - b_{r,1}} (b_{2} - b_{1}).$$
If $[b_1, b_2]$ and $[b_2, b_3]$ are neighbors,
\begin{eqnarray*}
\|b_{3} - b_{3}^*\| 
&\leq& 2^{j-J} \frac{\|b_{3} - b_{2}\| + \|b_{2} - b_{1}\|}{|b_{r_2,2} - b_{r_2,1}|} \\
&\leq& 2^{j-J} (\Delta/\delta_0 + 1) \\
&=& O(2^{-J} |B|^{-1}) 
\end{eqnarray*}
where the first inequality comes from (\ref{eq:beam-neigh}) and the second from the fact that for any beamlet $[b, b']$, $\delta \leq \|b' - b\| \leq \Delta$.
Now, in a ball of radius $A$ there are at most $O(A/\delta_0)^{d-1}$ $r$-gridpoints, for any $r = 1, \dots, d$.
Therefore, $B$ has $O(d |B|^{-(d-1)})$ neighbors. 
\QED

\subsection*{Proof of Lemma \ref{lem:beam-count}}
The number of $r$-gridpoints is of order $O(\Delta^{-1} \delta^{-(d-1)})$.
For a fixed $r$-gridpoint $b_1$, there are order $O(\Delta \delta^{-1})^{d-1}$ $r$-gridpoints $b_2$ such that $[b_1, b_2]$ forms an $r$-beam. 
Therefore, the number of $r$-beams is of order $O(\Delta^{d-2} \delta^{-2d+2})$.

The number of $r_1r_2$-beams is of smaller order of magnitude, because they are more constrained.
Indeed, $b_{r_1,2}$ and $b_{r_2,2}$ are determined by $b_{r_1,1}$ and $b_{r_2,1}$ up to $O(\delta)$, which corresponds to order $O(1)$ choices; while for each $b_{r_3,2}, r_3 \neq r_1, r_2$ there are order $O(\Delta \delta^{-1})$ choices.
Hence, there are order $O(\Delta^{d-3} \delta^{-2d+3})$ $r_1r_2$-beams.

All together, the number of beams is of order $O(d \Delta^{d-2} \delta^{-2d+2} (1 + d \Delta^{-1} \delta))$.\\

We now bound the degree of a vertex in $\overline{\bbB}_{j,J}^d$.
By construction, an $r_1$-beam is connected only with $r_1$-beams, and (possibly) $r_1r_2$-beams; similarly, an $r_1r_2$-beam is connected only with $r_1$- and $r_2$-beams, and (possibly) $r_3r_1$- and $r_2r_3$-beams.

Fix an $r_1r_2$-beam $[b_1,b_2]$.
First, take an $r_2$-beam $[b_2,b_3]$ and define 
$$b_3^* = b_{2} + \frac{b_{r_2,3} - b_{r_2,2}}{b_{r_2,2} - b_{r_2,1}} (b_{2} - b_{1}),$$
which is the intersection of the line $(b_1,b_2)$ with the $r_2$-hyperplane $b_3$ belongs to.
If $[b_1,b_2]$ and $[b_2,b_3]$ are neighbors,  
\begin{eqnarray*}
\|b_{3} - b_{3}^*\| 
&\leq& (11\delta/20) \frac{\|b_{3} - b_{2}\| + \|b_{2} - b_{1}\|}{|b_{r_2,2} - b_{r_2,1}|} \\
&\leq& (11\delta/20) (4 + O(\delta)) \\
&<& (5/2) \delta, \quad \text{for $\delta$ small enough}, 
\end{eqnarray*}
where the first inequality comes from (\ref{eq:beam-neigh}) and the second from (\ref{eq:h-beam})-(\ref{eq:hh-beam-3}).
Therefore, there are at most $5$ choices per coordinate of $b_3$ (except for $b_{r_2,3}$).
Hence, $[b_1,b_2]$ has at most $5^{d-1}$ neighbors that are $r_2$-beams.
Similarly, $[b_1,b_2]$ has at most $5^{d-1}$ neighbors that are $r_1$-beams.


Next, take an $r_2r_3$-beam $[b_2,b_3]$ and define 
$$b_1^* = b_{2} + \frac{b_{r_1,3} - b_{r_1,2}}{b_{r_1,2} - b_{r_1,1}} (b_{2} - b_{3}),$$
which is the intersection of the line $(b_2,b_3)$ with the $r_1$-hyperplane $b_1$ belongs to.
If $[b_1,b_2]$ and $[b_2,b_3]$ are neighbors, performing the corresponding computations we arrive at
\begin{equation}
\|b_{1} - b_{1}^*\| < (11\delta/20) (6 + O(\delta)) < (7/2) \delta, \quad \text{for $\delta$ small enough}.  \nonumber
\end{equation}
Note that $[b_2,b_1^*]$ is proportional to $[b_2,b_3]$, with a constant of proportionality between 1 and 2, and $[b_2,b_3]$ satisfies (\ref{eq:hh-beam-1})-(\ref{eq:hh-beam-3}) with $r_1$ replaced by $r_3$; this together with the bound above (applied with $r_4 = r_3$) and the triangle inequality implies:
$$|b_{r_3,1} - b_{r_3,2}| > |b_{r_2,1} - b_{r_2,2}| - 6\delta.$$
Turning things around, define 
$$b_3^* = b_{2} + \frac{b_{r_3,3} - b_{r_3,2}}{b_{r_3,2} - b_{r_3,1}} (b_{2} - b_{1}),$$
which is the intersection of the line $(b_1,b_2)$ with the $r_3$-hyperplane $b_3$ belongs to.
Again using (\ref{eq:beam-neigh}) together with properties (\ref{eq:h-beam})-(\ref{eq:hh-beam-3}), and the above inequality, we get 
\begin{equation}
\|b_{3} - b_{3}^*\| < (11\delta/20) (6 + O(\delta)) < (7/2) \delta, \quad \text{for $\delta$ small enough}.  \nonumber
\end{equation}
Therefore, there are at most $7$ choices per coordinate of $b_3$ (except for $b_{r_3,3}$).
Hence, $[b_1,b_2]$ has at most $7^{d-1}$ neighbors that are $r_2r_3$-beams.
Similarly, $[b_1,b_2]$ has at most $7^{d-1}$ neighbors that are $r_3r_1$-beams.

The reasoning is similar when $[b_1,b_2]$ is an $r_1$-beam.
In particular, an $r$-beam that does not make an angle close a 45$^o$ with the $r$-hyperplanes it connects only has $r$-beams as neighbors, at most $5^{d-1}$ on each side.
\QED

\subsection*{Proof of Theorem \ref{th:curve-covering}}
\label{proof:curve-covering}
Assume the conditions Theorem \ref{th:curve-covering} are satisfied, with the constant $K$ chosen as to make all the forthcoming appearances of $O(\kappa \Delta^\alpha)$ sufficiently small compared to $\delta$.
Note that under these conditions, both $\delta$, $\Delta$ and $\delta/\Delta$ are all decreasing functions of $J$.
We choose a smooth parametrization by arc-length of $\gamma$.
Let $\ell = \length(\gamma)$.
 
We say that $(u_1, \dots, u_d) \in \bbR^d$ is an $r$-vector if $|u_{r}| \geq |u_{r_1}|, \forall r_1$.
Without loss of generality, assume $\gamma'(0)$ is an $r$-vector with
$\gamma_{r}'(0) > 0$.
We may also assume that $\gamma(0)$ belongs to an $r$-hyperplane, for otherwise we
work with an extension of $\gamma$ that reaches an $r$-hyperplane in the direction
$-\gamma'(0)$.
We recursively define an increasing sequence of arclengths $\{s_i :0 \leq i \leq I\}$.
First, let $s_0 = 0$.
Suppose $s_{i-1}$ has been defined and assume,
without loss of generality, that $\gamma'(s_{i-1})$ is an $r$-vector; then, let 
$$s_{i} = \inf\{s > s_{i-1}: \gamma(s) \ \text{belongs to an $r$-hyperplane and } |\gamma_{r}(s) - \gamma_{r}(s_{i-1})| \geq \Delta \},$$
with the usual convention $\inf\{\emptyset\} = \infty$.
We may also assume that $s_{i} < \infty$, for otherwise we work with an extension of $\gamma$ that reaches an $r$-hyperplane (in the case above) in the direction $\gamma'(\ell)$.
If $s_{i} = \ell$, then let $I = i$ and stop the recursion.

Define $b_i$ to be the gridpoint closest to $\gamma(s_i)$, and $B_i = [b_i,b_{i+1}]$ the line-segment joining $b_i$ and $b_{i+1}$.

\begin{itemize}
\item {\bf Claim 1.} $I \leq \lambda \Delta^{-1} + 2$.
 
\item {\bf Claim 2.} For $i=0,\dots,I-1$, $B_i$ is a vertex in $\overline{\bbB}_{j,J}^d$.

\item {\bf Claim 3.} For $i=0,\dots,I-2$, $B_{i}$ and $B_{i+1}$ are neighbors
in $\overline{\bbB}_{j,J}^d$.

\item {\bf Claim 4.} For $i=0,\dots,I-1$, $\gamma([s_i,s_{i+1}]) \subset
R(B_i)$.
\end{itemize}

\noindent {\it Preliminaries.}
\begin{itemize}

\item 
For $\Delta$ small enough,  
\begin{equation}
\label{eq:sdiff}
\Delta \leq s_{i+1} - s_i < 3 d^{1/2} \Delta, \quad \forall i=0,\dots,I-1.
\end{equation}
{\it Proof.}
We start with the lower bound. 
A Taylor expansion and the fact that $\gamma$ is parametrized by arc gives
$$\|\gamma(s_{i+1}) - \gamma(s_i)\| \leq s_{i+1} - s_i,$$
and by construction, $\|\gamma(s_{i+1}) - \gamma(s_i)\| \geq \Delta$.

We now turn to proving the upper bound.
Suppose without loss of generality that $\gamma'(s_i)$ is an $r$-vector.
Then, by construction $|\gamma_{r}(s) - \gamma_{r}(s_i)| \leq 2 \Delta, \forall s \in [s_i,s_{i+1}]$, and $|\gamma_{r}'(s_i)| \geq d^{-1/2}$.
This, (\ref{eq:beam_taylor1}) and the triangle inequality imply
$$2 \Delta + \kappa (s-s_i)^\alpha - d^{-1/2} (s-s_i) \geq 0, \quad \forall s \in [s_i,s_{i+1}].$$
And for $\Delta < (3^\alpha d^{\alpha/2} \kappa)^{1/(\alpha-1)}$, the left handside is negative for $s-s_i$ replaced by $3 d^{1/2} \Delta$, so that $s_{i+1}-s_i < 3 d^{1/2} \Delta$.
\QED

\item For $i=1,\dots,I$, if $\gamma'(s_i)$ is an $r_1$-vector (say) and
$\gamma'(s_{i-1})$ is a $r_2$-vector (say), then 
\begin{equation}
\label{eq:diag_vect}
|\gamma_{r_1}'(s_{i})| - |\gamma_{r_2}'(s_{i})| \leq 4 \kappa (s_i - s_{i-1})^{\alpha-1}.
\end{equation}
{\it Proof.}
Because $|\gamma_{r_2}'(s_{i-1})| \geq |\gamma_{r_1}'(s_{i-1})|$, we have
$$|\gamma_{r_1}'(s_{i})| - |\gamma_{r_2}'(s_{i})| \leq |\gamma_{r_1}'(s_{i})| -
|\gamma_{r_1}'(s_{i-1})| - (|\gamma_{r_2}'(s_{i})| - |\gamma_{r_2}'(s_{i-1})|).$$
Using Lemma \ref{lem:beam_taylor_derivative}, this implies
$$|\gamma_{r_1}'(s_{i})| - |\gamma_{r_2}'(s_{i})| \leq 4 \kappa (s_i -
s_{i-1})^{\alpha-1}.$$
\QED

\item For $i=1,\dots,I$, 
\begin{equation}
\label{eq:b-triangle}
\|b_{i} - \gamma(s_{i})\| \leq \delta/2.
\end{equation}
{\it Proof.}  By construction.
\QED

\end{itemize}

\noindent {\it Proof of Claim 1.}
A straightforward consequence of the fact that $s_{i+1} - s_i \geq \Delta$ for all $i = 0, \dots, I-1$.
The `+2' comes from possibly extending the curve as described above.  \QED

\noindent {\it Proof of Claim 2.}
We first prove that $B_0$ is a beam.
Assume without loss of generality that $\gamma(s_0)$ is on an $r_1$-hyperplane and $\gamma'(s_0)$ is a
$r_1$-vector; we now show that $B_0$ is an $r_1$-beam.
Applying (\ref{eq:beam_taylor1}) and the fact that
$|\gamma_{r_2}'(s_0)| \leq |\gamma_{r_1}'(s_0)|$ we get
$$|\gamma_{r_2}(s_1) - \gamma_{r_2}(s_0)| \leq |\gamma_{r_1}'(s_0)| (s_1 - s_0) + \kappa
(s_1 - s_0)^\alpha.$$
Using (\ref{eq:beam_taylor1}) again, together with $|\gamma_{r_1}(s_1) -
\gamma_{r_1}(s_0)| = \Delta$, we have
$$|\gamma_{r_1}'(s_0)| (s_1 - s_0) \leq \Delta + \kappa (s_1 - s_0)^\alpha.$$
Therefore,
$$|\gamma_{r_2}(s_1) - \gamma_{r_2}(s_0)| \leq \Delta + 2 \kappa (s_1 - s_0)^\alpha.$$
On the right handside, we use (\ref{eq:sdiff}) to get $\kappa (s_1 - s_0)^\alpha = O(\kappa \Delta^\alpha)$.
All together,
$$|\gamma_{r_2}(s_1) - \gamma_{r_2}(s_0)| \leq \Delta + O(\kappa \Delta^\alpha);$$
similarly
$$|\gamma_{r_3}(s_1) - \gamma_{r_3}(s_0)| \leq \Delta + O(\kappa \Delta^\alpha).$$
This, together with the triangle inequality and (\ref{eq:b-triangle}), shows that
$$\|b_1 - b_0\| \leq \delta + \Delta + O(\kappa \Delta^\alpha) < 2 \delta +
\Delta,$$
which implies that $\|b_1 - b_0\| \leq \delta + \Delta$ since $\|b_1 - b_0\|$ is an integer multiple of $\delta$.  This proves that $B_0$ is an $r_1$-beam.

We next consider $i=1,\dots,I-1$ and prove that $B_i$ is a beam.
Assume without loss of generality that $\gamma'(s_i)$ is an $r_1$-vector.
If $\gamma'(s_{i-1})$ is an $r_1$-vector, then $\gamma(s_{i})$ is on an $r_1$-hyperplane
and the situation is as above.
Therefore, assume without loss of generality that $\gamma'(s_{i-1})$ is a $r_2$-vector; we
now show that $B_i$ is an $r_1r_2$-beam.
With (\ref{eq:beam_taylor1}), we get
$$|\ |\gamma_{r_1}(s_{i+1}) - \gamma_{r_1}(s_{i})| - |\gamma_{r_2}(s_{i+1}) -
\gamma_{r_2}(s_{i})|\ | \leq (|\gamma_{r_1}'(s_{i})| - |\gamma_{r_2}'(s_{i})|) (s_{i+1} -
s_i) + 2 \kappa (s_{i+1} - s_i)^\alpha;$$
together with (\ref{eq:diag_vect}), this implies
$$|\ |\gamma_{r_1}(s_{i+1}) - \gamma_{r_1}(s_{i})| - |\gamma_{r_2}(s_{i+1}) - \gamma_{r_2}(s_{i})|\ | \leq 6 \kappa (s_{i+1} - s_i)^\alpha,$$
and with (\ref{eq:sdiff}), 
$$|\ |\gamma_{r_1}(s_{i+1}) - \gamma_{r_1}(s_{i})| - |\gamma_{r_2}(s_{i+1}) - \gamma_{r_2}(s_{i})|\ | = O(\kappa \Delta^\alpha).$$
On the other hand, using (\ref{eq:beam_taylor1}) and the fact
that $|\gamma_{r_1}'(s_{i})| \geq |\gamma_{r_3}'(s_{i})|$, we also get
$$|\gamma_{r_1}(s_{i+1}) - \gamma_{r_1}(s_{i})| - |\gamma_{r_3}(s_{i+1}) - \gamma_{r_3}(s_{i})| \geq - 2 \kappa (s_{i+1} - s_i)^\alpha,$$
which by (\ref{eq:sdiff}) implies 
$$|\gamma_{r_1}(s_{i+1}) - \gamma_{r_1}(s_{i})| - |\gamma_{r_3}(s_{i+1}) - \gamma_{r_3}(s_{i})| \geq - O(\kappa \Delta^\alpha).$$

Using the equations above together with the triangle inequality and (\ref{eq:b-triangle}),
$$|\ |b_{r_1,i+1}-b_{r_1,i}| - |b_{r_2,i+1}-b_{r_2,i}|\ | \leq \delta + O(\kappa \Delta^\alpha) < 2\delta,$$
$$|b_{r_1,i+1}-b_{r_1,i}| - |b_{r_3,i+1}-b_{r_3,i}| \geq -  \delta - O(\kappa \Delta^\alpha) > -2 \delta.$$  
Since the differences above are multiple integers of $\delta$, the first inequality may be replaced by $\leq \delta$ and the second by $\geq - \delta$.
Therefore $B_i$ is an $r_1r_2$-beam. 
\QED

\noindent {\it Proof of Claim 3.}
Fix coordinate $r_1$; we want to show that:
\begin{eqnarray*}
&& \|(b_{r_1,i+1} - b_{r_1,i})(b_{i+2} - b_{i+1}) - (b_{r_1,i+2} - b_{r_1,i+1})(b_{i+1} - b_{i})\| \\
&& \qquad \qquad \qquad \qquad < (11\delta/20) (\|b_{i+2} - b_{i+2}\| + \|b_{i+1} - b_{i}\|).
\end{eqnarray*}

We first prove a similar inequality involving $\gamma$:
\begin{eqnarray*}
&& \|(\gamma_{r_1}(s_{i+1}) - \gamma_{r_1}(s_{i}))(\gamma(s_{i+2}) - \gamma(s_{i+1})) - (\gamma_{r_1}(s_{i+2}) - \gamma_{r_1}(s_{i+1})) (\gamma(s_{i+1}) - \gamma(s_i)) \| \\
&& \qquad \qquad \qquad \qquad \leq O(\kappa \Delta^\alpha) (\|\gamma(s_{i+2}) - \gamma(s_{i+1})\| + \|\gamma(s_{i+1}) - \gamma(s_{i})\|).
\end{eqnarray*}
This simply comes from applying (\ref{eq:beam_taylor1}) to get
\begin{align} 
\gamma_{r_1}(s_{i+2}) - \gamma_{r_1}(s_{i+1}) &= \gamma_{r_1}'(s_{i+1}) (s_{i+2}-s_{i+1}) + O(\kappa (s_{i+2}-s_{i+1})^\alpha), \nonumber \\
\gamma_{r_1}(s_{i+1}) - \gamma_{r_1}(s_{i}) &= \gamma_{r_1}'(s_{i+1}) (s_{i+1}-s_i) + O(\kappa (s_{i+1}-s_i)^\alpha), \nonumber \\
\gamma(s_{i+2}) - \gamma(s_{i+1}) &= \gamma'(s_{i+1}) (s_{i+2}-s_{i+1}) + O(\kappa (s_{i+2}-s_{i+1})^\alpha), \nonumber \\
\gamma(s_{i+1}) - \gamma(s_{i}) &= \gamma'(s_{i+1}) (s_{i+1}-s_i) + O(\kappa (s_{i+1}-s_i)^\alpha), \nonumber
\end{align}
and then using the triangle inequality and (\ref{eq:sdiff}).

Using this inequality, the triangle inequality and (\ref{eq:b-triangle}), we then get  
\begin{eqnarray*}
&& \|(b_{r_1,i+1} - b_{r_1,i})(b_{i+2} - b_{i+1}) - (b_{r_1,i+2} - b_{r_1,i+1})(b_{i+1} - b_{i})\| \\
&& \qquad \qquad \qquad \qquad < (\delta/2 + O(\kappa \Delta^\alpha)) (\|b_{i+2} - b_{i+2}\| + \|b_{i+1} - b_{i}\| + 2 \delta) \\
&& \qquad \qquad \qquad \qquad < (\delta/2 + O(\kappa \Delta^\alpha)) (1 + \Delta^{-1} \delta) (\|b_{i+2} - b_{i+2}\| + \|b_{i+1} - b_{i}\|),
\end{eqnarray*}
where we used the fact that any beam has (as a vector) supnorm at least $\Delta$.
We conclude by making $O(\kappa \Delta^\alpha)$ sufficiently small compared with $\delta$ and $\delta$ itself sufficiently small compared with $\Delta$.

Therefore, $B_i = [b_i,b_{i+1}]$ and $B_{i+1} = [b_{i+1},b_{i+2}]$ are neighbors in $\overline{\bbB}_{j,J}^d$.  \QED

\noindent {\it Proof of Claim 4.}
Applying Lemma \ref{lem:beam_taylor2} with $r = s_{i}$ and $t = s_{i+1}$, it
follows that $\gamma([s_{i},s_{i+1}])$ belongs to the $2 \kappa (s_{i+1}-
s_i)^\alpha$-neighborhood of $[\gamma(s_{i}),\gamma(s_{i+1})]$.
Because of (\ref{eq:sdiff}), (\ref{eq:b-triangle}) and the triangle inequality, this implies
that $\gamma([s_{i},s_{i+1}]) \subset R(B_i)$. \QED


\bibliographystyle{abbrv}
\addcontentsline{toc}{section}{\refname}
\bibliography{MinMaxDetect}

\def\cprime{$'$} \def\lfhook#1{\setbox0=\hbox{#1}{\ooalign{\hidewidth
  \lower1.5ex\hbox{'}\hidewidth\crcr\unhbox0}}}
\begin{thebibliography}{10}

\bibitem{sdss}
Sloan Digital Sky Survey (\url{http://www.sdss.org}).

\bibitem{leda}
Library of Efficient Data types and Algorithms
  (\url{http://www.mpi-inf.mpg.de/LEDA}).

\bibitem{code}
Software used in this paper (\url{http://www.cbl.uh.edu/~efros/research.html}).

\bibitem{longest-path-approx}
N.~Alon, R.~Yuster, and U.~Zwick.
\newblock Color-coding.
\newblock {\em J. Assoc. Comput. Mach.}, 42(4):844--856, 1995.

\bibitem{ery-thesis}
E.~Arias-Castro.
\newblock {\em Graphical structures for geometric detection}.
\newblock PhD thesis, Stanford University, 2004.

\bibitem{MGD}
E.~Arias-Castro, D.~Donoho, and X.~Huo.
\newblock Near-optimal detection of geometric objects by fast multiscale
  methods.
\newblock {\em IEEE Trans. Inform. Theory}, 51(7):2402--2425, 2005.

\bibitem{MSDFS}
E.~Arias-Castro, D.~Donoho, and X.~Huo.
\newblock Adaptive multiscale detection of filamentary structures in a
  background of uniform random points.
\newblock {\em Ann. Statist.}, 34(1):326--349, 2006.

\bibitem{CTD}
E.~Arias-Castro, D.~Donoho, X.~Huo, and C.~Tovey.
\newblock Connect the dots: how many random points can a regular curve pass
  through?
\newblock {\em Adv. in Appl. Probab.}, 37(3):571--603, 2005.

\bibitem{MR1092983}
R.~Arratia, L.~Goldstein, and L.~Gordon.
\newblock Poisson approximation and the {C}hen-{S}tein method.
\newblock {\em Statist. Sci.}, 5(4):403--434, 1990.
\newblock With comments and a rejoinder by the authors.

\bibitem{FSS}
A.~Averbuch, R.~Coifman, D.~Donoho, and M.~Israeli.
\newblock Fast slant stack: A notion of radon transform for data in a cartesian
  grid which is rapidly computible, algebraically exact, geometrically faithful
  and invertible.
\newblock Tech. report, Stanford University, 2001.

\bibitem{prize-collecting}
B.~Awerbuch, Y.~Azar, A.~Blum, and S.~Vempala.
\newblock New approximation guarantees for minimum-weight {$k$}-trees and
  prize-collecting salesmen.
\newblock {\em SIAM J. Comput.}, 28(1):254--262 (electronic), 1999.

\bibitem{1238310}
Y.~Boykov and V.~Kolmogorov.
\newblock Computing geodesics and minimal surfaces via graph cuts.
\newblock {\em Computer Vision, 2003. Proceedings. Ninth IEEE International
  Conference on}, pages 26--33 vol.1, Oct. 2003.

\bibitem{MR1857974}
E.~J. Candes.
\newblock Ridgelets and the representation of mutilated {S}obolev functions.
\newblock {\em SIAM J. Math. Anal.}, 33(2):347--368 (electronic), 2001.

\bibitem{chirp-tech}
E.~J. Cand\`es.
\newblock Multiscale chirplets and near-optimal recovery of chirps.
\newblock Technical report, Stanford University, 2002.

\bibitem{MR2379113}
E.~J. Cand{\`e}s, P.~R. Charlton, and H.~Helgason.
\newblock Detecting highly oscillatory signals by chirplet path pursuit.
\newblock {\em Appl. Comput. Harmon. Anal.}, 24(1):14--40, 2008.

\bibitem{MR2012649}
E.~J. Cand{\`e}s and D.~L. Donoho.
\newblock New tight frames of curvelets and optimal representations of objects
  with piecewise {$C\sp 2$} singularities.
\newblock {\em Comm. Pure Appl. Math.}, 57(2):219--266, 2004.

\bibitem{DasHesSon}
B.~DasGupta, J.~Hespanha, and E.~Sontag.
\newblock Computational complexities of honey-pot searching with local sensory
  information.
\newblock In {\em 2004 American Control Conference (ACC 2004)}, pages
  2134--2138, 2004.

\bibitem{766676}
A.~Desolneux, L.~Moisan, and J.-M. Morel.
\newblock A grouping principle and four applications.
\newblock {\em IEEE Trans. Pattern Anal. Mach. Intell.}, 25(4):508--513, 2003.

\bibitem{MR2036391}
A.~Desolneux, L.~Moisan, and J.-M. Morel.
\newblock Maximal meaningful events and applications to image analysis.
\newblock {\em Ann. Statist.}, 31(6):1822--1851, 2003.

\bibitem{spie-beamlets}
D.~Donoho, O.~Levi, J.-L. Starck, and V.~Martinez.
\newblock Multiscale geometric analysis for 3-d catalogues.
\newblock In J.-L. Starck and F.~Murtagh, editors, {\em SPIE Conference on
  Astronomical Data Analysis}, volume 4847, 2002.

\bibitem{MR1256530}
D.~L. Donoho.
\newblock Unconditional bases are optimal bases for data compression and for
  statistical estimation.
\newblock {\em Appl. Comput. Harmon. Anal.}, 1(1):100--115, 1993.

\bibitem{wedgelets}
D.~L. Donoho.
\newblock Wedgelets: nearly minimax estimation of edges.
\newblock {\em Ann. Statist.}, 27(3):859--897, 1999.

\bibitem{2D-beamlets}
D.~L. Donoho and X.~Huo.
\newblock Beamlets and multiscale image analysis.
\newblock In {\em Multiscale and multiresolution methods}, volume~20 of {\em
  Lect. Notes Comput. Sci. Eng.}, pages 149--196. Springer, Berlin, 2002.

\bibitem{MR1635414}
D.~L. Donoho and I.~M. Johnstone.
\newblock Minimax estimation via wavelet shrinkage.
\newblock {\em Ann. Statist.}, 26(3):879--921, 1998.

\bibitem{3D-beamlets}
D.~L. Donoho and O.~Levi.
\newblock Fast {X}-ray and beamlet transforms for three-dimensional data.
\newblock In {\em Modern signal processing}, volume~46 of {\em Math. Sci. Res.
  Inst. Publ.}, pages 79--116. Cambridge Univ. Press, Cambridge, 2004.

\bibitem{dp_tracking}
D.~Geiger, A.~Gupta, L.~Costa, and J.~Vlontzos.
\newblock Dynamic programming for detecting, tracking and matching deformable
  contours.
\newblock {\em IEEE Trans. on Pattern Analysis and Machine Intelligence},
  17(3):294--302, 1995.

\bibitem{GlaBal}
J.~Glaz and N.~Balakrishnan, editors.
\newblock {\em Scan statistics and applications}.
\newblock Statistics for Industry and Technology. Birkh\"auser Boston Inc.,
  Boston, MA, 1999.

\bibitem{GlaNauWal}
J.~Glaz, J.~Naus, and S.~Wallenstein.
\newblock {\em Scan statistics}.
\newblock Springer Series in Statistics. Springer-Verlag, New York, 2001.

\bibitem{Hu_optimalminimum-surface}
T.~C. Hu, A.~B.~K. T, and G.~Robins.
\newblock Optimal minimum-surface computations using network flow.
\newblock To appear in Mathematical Programming.

\bibitem{1518933}
X.~Huo and J.~Chen.
\newblock Jbeam: multiscale curve coding via beamlets.
\newblock {\em Image Processing, IEEE Transactions on}, 14(11):1665--1677, Nov.
  2005.

\bibitem{CTD-DP}
X.~Huo, D.~Donoho, C.~Tovey, and E.~Arias-Castro.
\newblock Dynamic programming methods for `connecting-the-dots' in scattered
  point sets.
\newblock In preparation.

\bibitem{johnstone1998feg}
I.~Johnstone.
\newblock Function estimation in gaussian noise sequence models.
\newblock Draft of a monograph, available at
  www-stat.stanford.edu/$~$imj/baseb.pdf.

\bibitem{longest-path-NP}
D.~Karger, R.~Motwani, and G.~D.~S. Ramkumar.
\newblock On approximating the longest path in a graph.
\newblock {\em Algorithmica}, 18(1):82--98, 1997.

\bibitem{kirsanov2004dgm}
D.~Kirsanov and S.~Gortler.
\newblock {A discrete global minimization algorithm for continuous variational
  problems}.
\newblock Technical report, Harvard University, Computer Science, 2004.

\bibitem{kolmogorov}
A.~N. Kolmogorov.
\newblock {\em Selected works of {A}. {N}. {K}olmogorov. {V}ol. {III}},
  volume~27 of {\em Mathematics and its Applications (Soviet Series)}.
\newblock Kluwer Academic Publishers Group, Dordrecht, 1993.
\newblock Information theory and the theory of algorithms. Edited by A. N.
  Shiryayev. Translated from the 1987 Russian original by A. B. Sossinsky.

\bibitem{MR0124720}
A.~N. Kolmogorov and V.~M. Tikhomirov.
\newblock {$\varepsilon $}-entropy and {$\varepsilon $}-capacity of sets in
  functional space.
\newblock {\em Amer. Math. Soc. Transl. (2)}, 17:277--364, 1961.

\bibitem{10.1109/ICCV.2005.252}
V.~Kolmogorov and Y.~Boykov.
\newblock What metrics can be approximated by geo-cuts, or global optimization
  of length/area and flux.
\newblock {\em iccv}, 1:564--571, 2005.

\bibitem{kor-tsy}
A.~P. Korostel{\"e}v and A.~B. Tsybakov.
\newblock {\em Minimax theory of image reconstruction}, volume~82 of {\em
  Lecture Notes in Statistics}.
\newblock Springer-Verlag, New York, 1993.

\bibitem{MR2128287}
E.~Le~Pennec and S.~Mallat.
\newblock Sparse geometric image representations with bandelets.
\newblock {\em IEEE Trans. Image Process.}, 14(4):423--438, 2005.

\bibitem{shas}
O.~Levi and B.~A. Efros.
\newblock A new fast algorithm for exact calculation of the discrete 2-d and
  3-d x-ray transform.
\newblock In {\em Lecture Series on Computer and Computational Sciences},
  volume~4, pages 319--322. ICCMSE, 2005.

\bibitem{IMA_MFD}
O.~Levi, S.~Rotman, and B.~Efros.
\newblock Multiframe dim target detection using 3d multiscale geometric
  analysis.
\newblock IMA Annual Program Year Workshop New Mathematics and Algorithms for
  3-D Image Analysis, January 2006.

\bibitem{649458}
M.~Lindenbaum and A.~Berengolts.
\newblock A probabilistic interpretation of the saliency network.
\newblock In {\em ECCV '00: Proceedings of the 6th European Conference on
  Computer Vision-Part II}, pages 257--272, London, UK, 2000. Springer-Verlag.

\bibitem{APPLIC_FILAMENT_IN_CELLS}
A.~Lyazghi, C.~Decaestecker, I.~Camby, R.~Kiss, and V.~Ham.
\newblock Characterization of actin filament in cancer cells by the hough
  transform.
\newblock In {\em Proceedings of the IASTED International Conference : Signal
  Processing, Pattern Recognition and Applications}, pages 183--142, Rhodes,
  Greece, July 2001.

\bibitem{MR1614527}
S.~Mallat.
\newblock {\em A wavelet tour of signal processing}.
\newblock Academic Press Inc., San Diego, CA, 1998.

\bibitem{MarSaa}
V.~Mart\'inez and E.~Saar.
\newblock {\em Statistics of the Galaxy Distribution}.
\newblock Chapman and Hall/CRC press, Boca Raton, 2002.

\bibitem{MR1228209}
Y.~Meyer.
\newblock {\em Wavelets and operators}, volume~37 of {\em Cambridge Studies in
  Advanced Mathematics}.
\newblock Cambridge University Press, Cambridge, 1992.
\newblock Translated from the 1990 French original by D. H. Salinger.

\bibitem{networks_structure}
M.~E.~J. Newman.
\newblock The structure and function of complex networks.
\newblock {\em SIAM Review}, 45:167--256, 2003.

\bibitem{water-quality}
G.~P. Patil, J.~Balbus, G.~Biging, J.~Jaja, W.~L. Myers, and C.~Taillie.
\newblock Multiscale advanced raster map analysis system: definition, design
  and development.
\newblock {\em Environ. Ecol. Stat.}, 11(2):113--138, 2004.

\bibitem{patil-upper}
G.~P. Patil and C.~Taillie.
\newblock Upper level set scan statistic for detecting arbitrarily shaped
  hotspots.
\newblock {\em Environ. Ecol. Stat.}, 11(2):183--197, 2004.

\bibitem{APPLIC_3D_NETWORK}
V.~Prinet, O.~Monga, and S.~Ma.
\newblock Extraction of vascular network in 3d images.
\newblock In {\em Proceedings: The IEEE International Conference on Image
  Processing}, pages 307--310, September 1996.

\bibitem{590008}
A.~Sha'asua and S.~Ullman.
\newblock Structural saliency: The detection of globally salient structures
  using a locally connected network.
\newblock {\em Computer Vision., Second International Conference on}, pages
  321--327, Dec 1988.

\bibitem{ShoWel}
G.~R. Shorack and J.~A. Wellner.
\newblock {\em Empirical processes with applications to statistics}.
\newblock Wiley Series in Probability and Mathematical Statistics: Probability
  and Mathematical Statistics. John Wiley \& Sons Inc., New York, 1986.

\bibitem{MR1045442}
G.~Wahba.
\newblock {\em Spline models for observational data}, volume~59 of {\em
  CBMS-NSF Regional Conference Series in Applied Mathematics}.
\newblock Society for Industrial and Applied Mathematics (SIAM), Philadelphia,
  PA, 1990.

\bibitem{network_social}
S.~Wasserman and K.~Faust.
\newblock {\em Social Network Analysis: Methods and Applications}.
\newblock Cambridge University Press, Cambridge, UK, 1994.

\bibitem{1199635}
R.~Willett and R.~Nowak.
\newblock Platelets: a multiscale approach for recovering edges and surfaces in
  photon-limited medical imaging.
\newblock {\em Medical Imaging, IEEE Transactions on}, 22(3):332--350, March
  2003.

\end{thebibliography}

\end{document}